%% file: main.tex
\documentclass{article}
\usepackage{algorithm}
\usepackage{subcaption}
\usepackage{graphicx}
\usepackage{algorithmic}
\usepackage[english]{babel}
\usepackage[utf8x]{inputenc}
\usepackage[T1]{fontenc}
\usepackage[a4paper,top=3cm,bottom=2cm,left=1.5cm,right=1.5cm,marginparwidth=1.5cm]{geometry}
\usepackage{amsmath}
\usepackage{booktabs} 
\usepackage{amssymb}
\usepackage{amsthm}
\usepackage{parskip}
\newtheorem{theorem}{Theorem}
\usepackage{thm-restate}
\newtheorem{lemma}{Lemma}
\newtheorem{fact}{Fact}
\newtheorem{definition}{Definition}
\newtheorem{corollary}{Corollary}
\newtheorem{proposition}{Proposition}
\newtheorem{remark}{Remark}
\usepackage{hyperref}
\hypersetup{
	colorlinks=true,
	linkcolor=blue!70!black,
	citecolor=blue!70!black,
	urlcolor=blue!70!black
}

\input{macros.tex}
\title{Private Geometric Median in Nearly-Linear Time}
\author{
Syamantak Kumar\thanks{University of Texas at Austin, \texttt{syamantak@utexas.edu}} \and 
Daogao Liu\thanks{Google Research, \texttt{liudaogao@gmail.com}} \and 
Kevin Tian\thanks{University of Texas at Austin, \texttt{kjtian@cs.utexas.edu}} \and 
Chutong Yang\thanks{University of Texas at Austin, \texttt{cyang98@utexas.edu}}
}
\date{}
\begin{document}

\maketitle

\begin{abstract}
Estimating the geometric median of a dataset is a robust counterpart to mean estimation, and is a fundamental problem in computational geometry. Recently, \cite{haghifam2024private} gave an $(\epsilon, \delta)$-differentially private algorithm obtaining an $\alpha$-multiplicative approximation to the geometric median objective, $\frac 1 n \sum_{i \in [n]} \|\cdot - \mathbf{x}_i\|$, given a dataset $\calD \defeq \{\mathbf{x}_i\}_{i \in [n]} \subset \mathbb{R}^d$. Their algorithm requires $n \gtrsim \sqrt d \cdot \frac 1 {\alpha\epsilon}$ samples, which they prove is information-theoretically optimal. This result is surprising because its error scales with the \emph{effective radius} of $\calD$ (i.e., of a ball capturing most points), rather than the worst-case radius. We give an improved algorithm that obtains the same approximation quality, also using $n \gtrsim \sqrt d \cdot \frac 1 {\alpha\epsilon}$ samples, but in time $\widetilde{O}(nd + \frac d {\alpha^2})$. Our runtime is nearly-linear, plus the cost of the cheapest non-private first-order method due to \cite{cohen2016geometric}. To achieve our results, we use subsampling and geometric aggregation tools inspired by FriendlyCore \cite{tsfadia2022friendlycore} to speed up the ``warm start'' component of the \cite{haghifam2024private} algorithm, combined with a careful custom analysis of DP-SGD's sensitivity for the geometric median objective.

\end{abstract}

\input{intro}
\input{prel}

\input{constant}

\input{boosting}

\input{experiments}

\section*{Acknowledgments}

We thank the Texas Advanced Computing Center (TACC) for computing resources used in this project.

\newpage

\bibliographystyle{alpha}
\bibliography{refs}

\newpage

\input{appendix}

\end{document}

%% file: macros.tex
\usepackage{fullpage}

\newcommand{\defeq}{:=}

\newcommand{\norm}[1]{\left\lVert#1\right\rVert}
\newcommand{\norms}[1]{\lVert#1\rVert}

\newcommand{\inprod}[2]{\left\langle#1, #2\right\rangle}
\newcommand{\inprods}[2]{\langle#1, #2\rangle}
\newcommand{\eps}{\epsilon}
\newcommand{\lam}{\lambda}
\newcommand{\sig}{\sigma}

\newcommand{\argmin}{\textup{argmin}} 
\newcommand{\R}{\mathbb{R}}

\newcommand{\N}{\mathbb{N}}

\newcommand{\half}{\frac{1}{2}}

\newcommand{\E}{\mathbb{E}}
\newcommand{\proj}{\boldsymbol{\Pi}}

\newcommand{\Nor}{\mathcal{N}}

\newcommand{\set}{\mathcal{K}}
\newcommand{\ball}{\mathbb{B}}
\newcommand{\id}{\mathbf{I}}

\newcommand{\dd}{\textup{d}}
\usepackage{xcolor}
\definecolor{burntorange}{rgb}{0.8, 0.33, 0.0}
\newcommand{\kjtian}[1]{\textcolor{burntorange}{\textbf{kjtian:} #1}}

\newcommand{\tO}{\widetilde{O}}

\newcommand{\Par}[1]{\left(#1\right)}
\newcommand{\Brack}[1]{\left[#1\right]}
\newcommand{\Brace}[1]{\left\{#1\right\}}
\newcommand{\bas}[1]{\begin{align*}#1\end{align*}}
\newcommand{\ba}[1]{\begin{align}#1\end{align}}
\newcommand{\bbb}[1]{\left[#1\right]}
\newcommand{\bb}[1]{\left(#1\right)}
\newcommand{\Abs}[1]{\left|#1\right|}

\newcommand{\alg}{\mathcal{A}}
\newcommand{\ind}{\mathbb{I}}

\newcommand{\1}{\mathbf{1}}
\newcommand{\0}{\mathbf{0}}

\newcommand{\vg}{\mathbf{g}}

\newcommand{\vu}{\mathbf{u}}
\newcommand{\vv}{\mathbf{v}}

\newcommand{\vx}{\mathbf{x}}
\newcommand{\vy}{\mathbf{y}}
\newcommand{\vz}{\mathbf{z}}
\newcommand{\vxi}{\boldsymbol{\xi}}
\newcommand{\vmu}{\boldsymbol{\mu}}

\newcommand{\calA}{\mathcal{A}}

\newcommand{\calD}{\mathcal{D}}
\newcommand{\calE}{\mathcal{E}}

\newcommand{\calI}{\mathcal{I}}

\newcommand{\calN}{\mathcal{N}}

\newcommand{\calX}{\mathcal{X}}

\usepackage{bbold}

\newcommand{\qcon}{s^{\mathsf{conc}}}

\newcommand{\qcont}{\tilde{s}^{\mathsf{conc}}}

\newcommand{\AboTh}{\mathsf{AboveThreshold}} 

\newcommand{\reals}{\mathbb{R}} 

\newcommand{\Unif}{\mathsf{Unif}} 

\newcommand{\Prob}{\mathbb{P}} 

\newcommand{\Lap}{\mathsf{Laplace}} 

\newcommand{\thresh}{\mathsf{thresh}}
\newcommand{\vxs}{\vx_\star}
\newcommand{\fcd}{f_{\calD}}
\newcommand{\event}{\calE}
\newcommand{\xset}{\calX}
\newcommand{\bvx}{\bar{\vx}}
\newcommand{\Bern}{\mathsf{Bern}}
\newcommand{\hmu}{\hat{\mu}}
\newcommand{\FRF}{\mathsf{FastRadius}}

\newcommand{\htau}{\hat{\tau}}
\newcommand{\hr}{\hat{r}}
\newcommand{\FC}{\mathsf{FastCenter}}
\newcommand{\BL}{\mathsf{BoundedLaplace}}
\newcommand{\vtheta}{\pmb{\theta}}
\newcommand{\hvx}{\hat{\vx}}
\newcommand{\balg}{\overline{\mathcal{A}}}

\newcommand{\StableDPSGD}{\mathsf{StableDPSGD}}

\newcommand{\DPGD}{\mathsf{DPGD}}
\newcommand{\FixedOrderDPSGD}{\mathsf{FixedOrderDPSGD}}

\newcommand{\tvg}{\tilde{\vg}}
\newcommand{\gaussiancluster}{\mathsf{GaussianCluster}}
\newcommand{\heavytailed}{\mathsf{HeavyTailed}}

%% file: intro.tex
\section{Introduction}\label{sec:introduction}

The \emph{geometric median} problem, also known as the Fermat-Weber problem, is one of the oldest problems in computational geometry. In this problem, we are given a dataset $\calD=\{\vx_i\}_{i\in [n]}\subset \R^d$, and our goal is to 
find a point $\vxs \in \R^d$ that minimizes the average Euclidean distance to points in the dataset:
\begin{equation}\label{eq:gm_intro}
    \vxs \in \arg\min_{\vx \in\R^d} f_\calD(\vx), \text{ where } f_\calD(\vx) \defeq \frac 1 n \sum_{i\in[n]}\|\vx -\vx_i\|.
\end{equation}
This problem has received widespread interest due to its applications in high-dimensional statistics. In particular, the geometric median of a dataset $\calD$ enjoys robustness properties that the mean (i.e., $\frac 1 n \sum_{i \in [n]} \vx_i$, the minimizer of $\frac 1 n \sum_{i \in [n]} \norm{\vx - \vx_i}^2$) does not. For example, it is known (cf.\ Lemma~\ref{lm:property_of_gm}) that if greater than half of $\calD$ lies within a distance $r$ of some $\bvx \in \R^d$, then the geometric median lies within $O(r)$ of $\bvx$. Thus, the geometric median provides strong estimation guarantees even when $\calD$ contains outliers. This is in contrast to simpler estimators such as the mean, which can be arbitrarily corrupted by a single outlier. As a result, studying the properties and computational aspects of the geometric median has a long history, see e.g., \cite{weber1929theory, LopuhaaR91} for some famous examples. 

In this paper, we provide improved algorithms for estimating \eqref{eq:gm_intro} subject to $(\eps, \delta)$-differential privacy (DP, Definition~\ref{def:dp}), the de facto notion of provable privacy in modern machine learning. Privately computing the geometric median naturally fits into a recent line of work on designing DP algorithms in the presence of outliers. To explain the challenge of such problems, the definition of DP implies that the privacy-preserving guarantee must hold for \emph{worst-case} datasets. This stringent definition affords DP a variety of desirable properties, most notably \emph{composition} of private mechanisms (cf.\ \cite{dr14}, Section 3.5). However, it also begets challenges: for example, estimating the empirical mean of $\calD$ subject to $(\eps,\delta)$-DP necessarily results in error scaling $\propto R$, the diameter of the dataset (cf.\ Section 5, \cite{BassilyST14}). Moreover, the worst-case nature of DP is at odds with typical \emph{average-case} machine learning settings, where most (or all) of $\calD$ is drawn from a distribution that we wish to learn about. From an algorithm design standpoint, the question follows: how do we design methods that provide privacy guarantees for worst-case data, but also yield improved utility guarantees for (mostly) average-case data? 

Such questions have been successfully addressed for various statistical tasks in recent work, including parameter estimation \cite{barber2014privacy, karwa2017finite, bun2019private, du2020differentially, BiswasDKU20, brown2021covariance, ashtiani2022private, LiuKJO22, kuditipudi2023pretty, BrownHS23}, clustering \cite{NissimRS07, NissimSV16, CohenKMST21, tsfadia2022friendlycore}, and more. However, existing approaches for estimating \eqref{eq:gm_intro} (even non-privately) are based on iterative optimization methods, as the geometric median does not admit a simple, closed-form solution. Much of the DP optimization toolkit is exactly plagued by the aforementioned ``worst-case sensitivity'' issues, e.g., lower bounds for general stochastic optimization problems again scale with the domain size. This is troubling in the context of \eqref{eq:gm_intro}, because a major appeal of the geometric median is its robustness: its error should not be significantly affected by any small subset of the data. Privately estimating the geometric median thus poses an interesting technical challenge, beyond its potential appeal as a subroutine in downstream robust algorithms.

To explain the distinction between worst-case and average-case error rates in the context of \eqref{eq:gm_intro}, we introduce the following helpful notation: for all quantiles $\tau \in [0, 1]$, we let
\begin{equation}\label{eq:quantile_intro}r^{(\tau)} \defeq \arg\min_{r \ge 0}\Brace{\sum_{i \in [n]} \ind_{\norm{\vx_i - \vxs} \le r} \ge \tau n}, \text{ where } \vxs \defeq \arg\min_{\vx \in \R^d} \frac 1 n \sum_{i \in [n]} \norm{\vx - \vx_i}, \end{equation}
when $\calD = \{\vx_i\}_{i \in [n]} \subset \R^d$ is clear from context. In other words, $r^{(\tau)}$ is the smallest radius describing a ball around the geometric median $\vxs$ containing at least $\tau n$ points in $\calD$. We also use $R$ to denote an a priori overall domain size bound, where we are guaranteed that $\calD \subset \ball^d(R)$. Note that in general, it is possible for, e.g., $r^{(0.9)} \ll R$ if $\approx 10\%$ of $\calD$ consists of outliers with atypical norms. Due to the robust nature of the geometric median (i.e., the aforementioned Lemma~\ref{lm:property_of_gm}), a natural target is estimation error scaling with the ``effective radius'' $r^{(\tau)}$ for some quantile $\tau \in (0.5, 1)$. This is a much stronger guarantee than the error rates $\propto R$ that typical DP optimization methods  give.

Because a simple argument (Lemma~\ref{lem:opt_lb}) shows that $r^{(\tau)} = O(f_{\calD}(\vxs))$ for all $\tau < 1$, in this introduction our goal will be to approximate the minimizer of \eqref{eq:gm_intro} to additive error $\alpha f_{\calD}(\vxs)$ for some $\alpha \in (0, 1)$, i.e., to give $\alpha$-\emph{multiplicative error} guarantees on optimizing $f_{\calD}$.\footnote{Our results, as well as those of \cite{haghifam2024private}, in fact give stronger additive error bounds of $\alpha r^{(\tau)}$ for any fixed $\tau \in (0.5, 1)$.} Again, datasets with outliers may have $f_{\calD}(\vxs) \ll R$, so this goal is beyond the reach of na\"ively applying DP optimization methods.

In a recent exciting work, \cite{haghifam2024private} bypassed this obstacle and obtained such private multiplicative approximations to the geometric median, and with near-optimal sample complexity. Assuming that $\calD$ has size $n \gtrsim \sqrt{d} \cdot \frac 1 {\alpha\eps}$,\footnote{In this introduction only, we use $\tO, \lesssim, \gtrsim$ to hide polylogarithmic factors in problem parameters, i.e., $d$, $\frac 1 \alpha$, $\frac 1 \eps$, $\frac 1 \delta$, and $\frac R r$, where $\calD \subseteq \ball^d(R)$ and $r \le r^{(0.9)}$. Our formal theorem statements explicitly state our dependences on all parameters.} \cite{haghifam2024private} gave two algorithms for estimating \eqref{eq:gm_intro} to $\alpha$-multiplicative error (cf.\ Appendix~\ref{app:hsu_discuss}). They also proved a matching lower bound, showing that this many samples is information-theoretically necessary.\footnote{Intuitively, we require $\alpha \approx d^{-1/2}$ to obtain nontrivial mean estimation when $\calD$ consists of i.i.d.\ Gaussian data (as a typical radius is $\approx \sqrt d$), matching known sample complexity lower bounds of $\approx \frac d \eps$ for Gaussian mean estimation \cite{KamathLSU19}.}
From both a theoretical and practical perspective, the main outstanding question left by \cite{haghifam2024private} is that of computational efficiency: in particular, the \cite{haghifam2024private} algorithms ran in time $\tO(n^2 d + n^3 \eps^2)$ or $\tO(n^2 d + nd^2 + d^{4.372})$. This leaves a significant gap between algorithms for privately solving \eqref{eq:gm_intro}, and their counterparts in the non-private setting, where \cite{cohen2016geometric} showed that \eqref{eq:gm_intro} could be approximated to $\alpha$-multiplicative error in nearly-linear time $\tO(\min(nd, \frac d {\alpha^2}))$.

\subsection{Our results}\label{ssec:results}

Our main contribution is a faster algorithm for privately approximating \eqref{eq:gm_intro} to $\alpha$-multiplicative error.

\begin{theorem}[informal, see Theorem~\ref{thm:boost}]\label{thm:intro_boost}
Let $\calD = \{\vx_i\}_{i \in [n]} \subset \ball^d(R)$ for $R > 0$, $0 < r \le r^{(0.9)}$, and $(\alpha, \eps, \delta) \in [0, 1]^3$. There is an $(\eps, \delta)$-DP algorithm that returns $\hvx$ such that with probability $\ge 1 - \delta$, $f_{\calD}(\hvx) \le (1 + \alpha)f_{\calD}(\vxs)$,
assuming $n \gtrsim \frac{\sqrt d}{\alpha\eps}$. The algorithm runs in time $\tO(nd + \frac d {\alpha^2})$.
\end{theorem}

To briefly explain Theorem~\ref{thm:intro_boost}'s statement, it uses a priori knowledge of $0 < r < R$ such that $R$ upper bounds the domain size of $\calD$, and $r$ lower bounds the ``effective radius'' $r^{(0.9)}$. However, its runtime only depends polylogarithmically on the aspect ratio $\frac R r$, rather than polynomially (as na\"ive DP optimization methods would); we also remark that our sample complexity is independent of $\frac R r$. 

The runtime of Theorem~\ref{thm:intro_boost} is nearly-linear in the regime $n \gtrsim \frac 1 {\alpha^2}$ (e.g., if $\sqrt d \cdot \frac 1 \eps \gtrsim \frac 1 \alpha$), but more generally it does incur an additive overhead of $\frac d {\alpha^2}$. This overhead matches the fastest non-private first-order method for approximating \eqref{eq:gm_intro} to $\alpha$-multiplicative error, due to \cite{cohen2016geometric}. We note that \cite{cohen2016geometric} also gave a custom second-order interior-point method, that non-privately solves \eqref{eq:gm_intro} in time $\tO(nd)$, i.e., with polylogarithmic dependence on $\frac 1 \alpha$. We leave removing this additive runtime term in the DP setting, or proving this is impossible in concrete query models, as a challenging question for future work.

Our algorithm follows a roadmap given by \cite{haghifam2024private}, who split their algorithm into two phases: an initial ``warm start'' phase that computes an $O(1)$-multiplicative approximation of the geometric median, and a secondary ``boosting'' phase that uses iterative optimization methods to improve the warm start to an $\alpha$-multiplicative approximation. The role of the warm start is to improve the domain size of the boosting phase to scale with the effective radius. However, both the warm start and the boosting phases of \cite{haghifam2024private} required superlinear $\approx n^2 d$ time. Our improvement to the warm start phase of the \cite{haghifam2024private} is quite simple, and may be of independent interest, so we provide a self-contained statement here.

\begin{theorem}[informal, see Theorem~\ref{thm:constant_factor}]
Let $\calD = \{\vx_i\}_{i \in [n]} \subset \ball^d(R)$ for $R > 0$, $0 < r \le r^{(0.9)}$, and $(\eps, \delta) \in [0, 1]^2$. There is an $(\eps, \delta)$-DP algorithm that returns $\hvx$ such that with probability $\ge 1 - \delta$,
$f_{\calD}(\hvx)  = O(f_{\calD}(\vxs))$, 
assuming $n \gtrsim \frac{\sqrt d}{\eps}$. The algorithm runs in time $\tO(nd)$.
\end{theorem}

\subsection{Our techniques}

As discussed previously, our algorithm employs a similar framework as \cite{haghifam2024private}. It is convenient to further split the warm start phase of the algorithm into two parts: finding an estimate $\hr$ of the effective radius of $\calD$, and finding an approximate centerpoint at distance $O(\hr)$ from the geometric median $\vxs$.

\paragraph{Radius estimation.} Our radius estimation algorithm is almost identical to that in \cite{haghifam2024private}, Section 2.1, which uses the sparse vector technique (cf.\ Lemma~\ref{thm:Above_Threshold}) to detect the first time an estimate $\hr$ is such that most points have $\ge \frac 3 4$ of $\calD$ at a distance of $\approx \hr$. The estimate $\hr$ is geometrically updated over a grid of size $O(\log(\frac R r))$. Na\"ively implemented, this strategy takes $\gtrsim n^2 d$ time due to the need for pairwise distance comparisons (cf.\ Appendix~\ref{app:hsu_discuss}); even if dimesionality reduction techniques are used, this step appears to require $\Omega(n^2)$ time. We make a simple observation that a random sample of $\approx \log(\frac 1 \delta)$ points from $\calD$ is enough to determine whether a given point has $\gg \beta$ neighbors, or $\ll \gamma$, for appropriate (constant) quantile thresholds $\beta, \gamma$, which is enough to obtain an $\tO(nd)$ runtime.

\paragraph{Centerpoint estimation.} Our centerpoint estimation step departs from \cite{haghifam2024private}, Section 2.2, who analyzed a custom variant of DP gradient descent with geometrically-decaying step sizes. We make the simple observation that directly applying the FriendlyCore algorithm of \cite{haghifam2024private} yields the same result. However, the standard implementation of FriendlyCore again requires $\Omega(n^2)$ time to estimate weights for each data point. We again show that FriendlyCore can be sped up to run in $\tO(nd)$ time (independently of $\frac R r$) via weights estimated through subsampling. Our privacy proof of this subsampled variant is subtle, and based on an argument (Lemma~\ref{lem:dp_close}) that couples our algorithm to an idealized algorithm that never fails to be private. We use this to account for the privacy loss due to the failure of our subsampling, i.e., if the estimates are inaccurate. We note that the \cite{haghifam2024private} algorithm for this step already ran in nearly-linear $\approx nd\log(\frac R r)$ time, so we obtain an asymptotic improvement only if $\frac R r$ is large.

\paragraph{Boosting.} The most technically novel part of our algorithm is in the boosting phase, which takes as input a radius and centerpoint estimate from the previous steps, and outputs an $\alpha$-multiplicative approximation to \eqref{eq:gm_intro}. Like \cite{haghifam2024private}, we use iterative optimization methods to implement this phase. However, a major bottleneck to a faster algorithm is the lack of a nearly-linear time DP solver for non-smooth empirical risk minimization (ERM) problems. Indeed, such $\tO(1)$-pass optimizers are known only when the objective is convex and sufficiently smooth \cite{FeldmanKT20}, or $n \gtrsim d^2$ samples are taken \cite{CarmonJJLLST23}.
This is an issue, because while computing the geometric median \eqref{eq:gm_intro} is a convex ERM problem, it is non-smooth, and nontrivial multiplicative guarantees are possible even with $n \approx \sqrt d$ samples.

We give a custom analysis of DP-SGD, specifically catered to the (non-smooth) ERM objective \eqref{eq:gm_intro}. Our main contribution is a tighter sensitivity analysis of DP-SGD's iterates, leveraging the structure of the geometric median. To motivate this observation, consider coupled algorithms with iterates $\vz$, $\vz'$, both taking gradient steps with respect to the subsampled function $\norm{\cdot - \vx_i}$ for some dataset element $\vx_i \in \calD$. A simple calculation \eqref{eq:sample_grad} shows these gradients are unit vectors $\vu$, $\vu'$, in the directions of $\vz - \vx_i$ and $\vz' - \vx_i$ respectively. It is not hard to formalize (Lemma~\ref{lem:same_sgd_step}) that updating $\vz \gets \vz - \eta \vu$ and $\vz' \gets \vz' - \eta \vu'$ is always contractive, unless $\vz, \vz'$ were both already very close to $\vx_i$ (and hence, each other) to begin with. We use this structural result to inductively control DP-SGD's sensitivity, which lets us leverage a prior reduction from private optimization to stable optimization \cite{FeldmanKT20}.

Our result is the first we are aware of that obtains a nearly-linear runtime for DP-SGD on a structured non-smooth problem. We were inspired by \cite{AsiLT24}, who also gave faster runtimes for (smooth) DP optimization problems with outliers under further assumptions on the objective. We hope that our work motivates future DP optimization methods that harness problem structure for improved rates.

\subsection{Related work}\label{ssec:related}

\paragraph{Differentially private convex optimization.}
Differentially private convex optimization has been studied extensively for over a decade \cite{chaudhuri2008privacy,kifer2012private,BassilyST14,kasiviswanathan2016efficient,bassily2020stability,FeldmanKT20,bassily2021non,gopi2022private,gopi2023private} and inspired the influential DP-SGD algorithm widely adopted in deep learning \cite{abadi2016deep}.
In the classic setting, where functions are assumed to be Lipschitz and defined over a convex domain of diameter $R$, optimal rates have been achieved with linear dependence on $R$ \cite{bassily2019private}. Recent years have seen significant advancements in optimizing the gradient complexity of DP stochastic convex optimization \cite{FeldmanKT20,asi2021private,kulkarni2021private,zhang2022differentially,CarmonJJLLST23,choquette2024optimal}. Despite these efforts, a nearly-linear gradient complexity has only been established for sufficiently smooth functions \cite{FeldmanKT20,zhang2022differentially,choquette2024optimal} and for non-smooth functions \cite{CarmonJJLLST23} when the condition $\sqrt{n}\gtrsim d$ is satisfied.

\paragraph{Differential privacy with average-case data.}
Adapting noise to the inherent properties of data, rather than catering to worst-case scenarios, is critical for making differential privacy practical in real-world applications.
Several important approaches have emerged in this direction: smooth sensitivity frameworks \cite{NissimRS07} that refine local sensitivity to make it private; instance optimality techniques \cite{asi2020instance} that provide tailored guarantees for specific datasets; methods with improved performance under distributional assumptions such as sub-Gaussian or heavy-tailed i.i.d.\ data \cite{cai2021cost,asi2023user,AsiLT24}; and data-dependent sensitivity computations that adapt during algorithm execution \cite{andrew2021differentially}. These approaches collectively represent the frontier in balancing privacy and utility beyond worst-case analyses. We view our work as another contribution towards this broader program.

%% file: prel.tex
\section{Preliminaries }\label{sec:prelims}

In this section, we collect preliminary results used throughout the paper. We define our notation in Section~\ref{ssec:notation}. We then formally state helper definitions and known tools from the literature on differential privacy and computing the geometric median in Sections~\ref{ssec:dp} and~\ref{ssec:median} respectively.

\subsection{Notation and probability basics}\label{ssec:notation}

Throughout, vectors are denoted in lowercase boldface, and the all-zeroes and all-ones vectors in dimension $d$ are respectively denoted $\0_d$ and $\1_d$. We use $\norm{\cdot}$ to denote the Euclidean ($\ell_2$) norm of a vector argument. We use $[d]$ to denote $\{i \in \N \mid 1 \le i \le d\}$. We use $\ball^d(\vmu, r) \defeq \{\vx \in \R^d \mid \norm{\vx - \vmu} \le r\}$ to denote the Euclidean ball of radius $r > 0$ around $\vmu \in \R^d$; when $\vmu$ is unspecified, then $\vmu = \0_d$ by default. For a compact set $\set \subseteq \R^d$, we use $\proj_\set(\vx)$ to denote the Euclidean projection $\arg\min_{\vy \in \set}\norm{\vx - \vy}$.

We let $\ind_\event$ denote the $0$-$1$ indicator random variable corresponding to an event $\calE$. For two densities $\mu, \nu$ on the same probability space $\Omega$ and $\alpha > 1$, we define the $\alpha$-R\'enyi divergence by:
\[D_\alpha(\mu\|\nu) \defeq \frac 1 {\alpha - 1}\log\Par{\int \Par{\frac{\mu(\omega)}{\nu(\omega)}}^\alpha\nu(\omega)\dd\omega}.\]

We use $\Nor(\vmu, \sigma^2 \id_d)$ to denote the multivariate normal distribution with mean $\vmu \in \R^d$ and covariance $\sigma^2 \id_d$, where $\id_d$ denotes the $d \times d$ identity matrix. We let $\Lap(\lam)$ be the Laplace distribution with scale parameter $\lam \ge 0$, whose density is $\propto \exp(-\frac{|\cdot|}{\lam})$. We let $\Unif(S)$ denote the uniform distribution over a set $S$, and $\Bern(p)$ denote the Bernoulli distribution taking on values $\{0, 1\}$ with mean $p \in [0, 1]$. We refer to a product distribution consisting of $k$ i.i.d.\ copies of a base distribution $\calD$ by $\calD^{\otimes k}$.
We also use the bounded Laplace distribution with parameters $\lam, \tau \ge 0$, denoted $\BL(\lam, \tau)$, which is the distribution of $X \sim \Lap(\lam)$ conditioned on $|X| \le \tau$. 

Finally, we require the following standard bound on binomial concentration.

\begin{fact}[Chernoff bound]\label{fact:chernoff}
For all $i \in [n]$, let $X_i \sim \Bern(p_i)$ for some $p_i \in [0, 1]$, and let $\mu \defeq \sum_{i \in [n]} p_i$ and $\hmu \defeq \sum_{i \in [n]} X_i$. Then, 
\begin{align*}
\Pr\Brack{\hmu > (1 + \eps)\mu} \le \exp\Par{-\frac{\eps^2 \mu}{2+\eps}} \text{ for all } \eps \ge 0,\\ 
\Pr\Brack{\hmu < (1 - \eps)\mu} \le \exp\Par{-\frac{\eps^2 \mu}{2}} \text{ for all } \eps \in (0, 1).
\end{align*}
\end{fact}

\subsection{Differential privacy}\label{ssec:dp}

Let $\xset$ be some domain, and let $\calD \in \xset^n$ be a dataset consisting of $n$ elements from $\xset$. We say that two datasets $\calD$, $\calD' \in \xset^n$ are \emph{neighboring} if their symmetric difference has size $1$, i.e., they differ in a single element. We use the following definition of differential privacy in this paper.

\begin{definition}[Differential privacy]\label{def:dp}
Let $(\eps, \delta) \in [0, 1]^2$.\footnote{In principle, the privacy parameter $\eps$ can be larger than $1$. However, in this paper, sample complexities are unaffected up to constants for any $\eps \ge 1$ if we simply obtain $(1, \delta)$-DP guarantees rather than $(\eps, \delta)$-DP guarantees, which are only stronger. Thus we state all results for $\eps \in [0, 1]$ for convenience, which simplifies some bounds.}
We say that a randomized algorithm $\alg: \xset^n \to \Omega$ satisfies \emph{$(\eps, \delta)$-differential privacy} (or, is \emph{$(\eps, \delta)$-DP}) if for all events $\event \subseteq \Omega$, and for all neighboring datasets $\calD, \calD' \in \xset^n$, we have
\[\Pr\Brack{\alg(\calD) \in \event} \le \exp(\eps) \Pr\Brack{\alg(\calD') \in \event} + \delta.\]
\end{definition}

DP algorithms satisfy \emph{basic composition} (Theorem B.1, \cite{dr14}), i.e., if $\alg_1: \xset^n \to \Omega_1$ is $(\eps_1, \delta_1)$-DP and $\alg_2: \xset^n \times \Omega_1 \to \Omega_2$ is $(\eps_2, \delta_2)$-DP, then running $\alg_2$ on $\calD$ and the output of $\alg_1(\calD)$ is $(\eps_1 + \eps_2, \delta_1 + \delta_2)$-DP. We next state the Gaussian mechanism. Recall that if $\vv: \xset^n \to \R^k$ is a vector-valued function of a dataset, we say $\vv$ has sensitivity $\Delta$ if for all neighboring $\calD, \calD' \in \xset^n$, we have $\norm{\vv(\calD) - \vv(\calD')} \le \Delta$.

\begin{fact}[Theorem A.1, \cite{dr14}]\label{fact:gaussian_mech}
Let $\vv: \xset^n \to \R^k$ have sensitivity $\Delta$, and let $(\eps, \delta) \in [0, 1]^2$. Then, drawing a sample from $\Nor(\vv(\calD), \sigma^2 \id_k)$ is $(\eps, \delta)$-DP, for any $\sig \ge \frac{2\Delta}{\eps} \cdot \sqrt{\log(\frac 2 \delta)} $.
\end{fact}

We also require the bounded Laplace mechanism, which is known to give the following guarantee.

\begin{fact}[Lemma 9, \cite{AsiLT24}]\label{fact:bl_mech}
Let $s: \xset^n \to \R$ have sensitivity $\Delta$, and let $(\eps, \delta) \in [0, 1]^2$. Then, drawing $\xi \sim \BL(\frac \Delta \eps, \tau)$ and outputting $s(\calD) + \xi$ is $(\eps, \delta)$-DP for any $\tau \ge \frac \Delta \eps \log(\frac 4 \delta)$.
\end{fact}

Fact~\ref{fact:bl_mech} is proven in \cite{AsiLT24} using a coupling argument, using the fact that $\BL(\lam)$ and $\Lap(\lam)$ result in the same sample except with some probability. We appeal to this privacy proof technique several times in Section~\ref{sec:constant}, so we state it explicitly here for convenience.

\begin{lemma}\label{lem:dp_close}
For $(\eps, \delta) \in [0, 1]^2$, $\alg: \xset^n \to \Omega$ be an $(\eps, \delta)$-DP algorithm, and let $\balg$ be an algorithm such that on any input $\calD \in \xset^n$, we have that the total variation distance between $\alg(\calD)$ and $\balg(\calD)$ is at most $\delta'$. Then, $\balg$ is an $(\eps, \delta + 4\delta')$-DP algorithm.
\end{lemma}
\begin{proof}
For neighboring datasets $\calD, \calD'$, and any event $\event \in \Omega$, we have that
\begin{align*}
\Pr\Brack{\balg(\calD) \in \event} &\le \Pr\Brack{\alg(\calD) \in \event} + \delta' \\
&\le \exp\Par{\eps} \Pr\Brack{\alg(\calD') \in \event} + \delta + \delta' \\
&\le \exp\Par{\eps} \Pr\Brack{\balg(\calD') \in \event} + \delta + 4\delta'.
\end{align*}
The first and last lines used the assumption between $\calA$ and $\balg$, and the second line used that $\alg$ is DP.
\end{proof}

We next recall the following well-known result on detecting the first large element in a stream.

\begin{algorithm}
\caption{$\AboTh(\calD, \{q_t\}_{t \in [T]}, \tau, \eps)$}
\label{alg:mean_est_with_AT}
{\bf Input:} Dataset $\calD \in \xset^n$, sensitivity-$\Delta$ queries $\{q_t: \xset^n \to \R\}_{t \in [T]}$, threshold $\tau \in \reals$, privacy parameter $\epsilon > 0$\;
\begin{algorithmic}[1]
\STATE $\hat{\tau} \gets \tau + \nu_{\thresh}$ for $\nu_{\thresh} \sim \Lap(\frac {2\Delta} \eps)$\;
\FOR{$t \in [T]$}
\STATE $\nu_t \sim \Lap(\frac{4\Delta}{\epsilon})$\;
\IF{$q_t(\calD)+\nu_t \geq \hat{\tau}$}
\STATE {\bf Output:} $a_t \gets \top$\;
\STATE {\bf Halt}\;
\ELSE
\STATE {\bf Output:} $a_t \gets \bot$\;

\ENDIF
\ENDFOR
\end{algorithmic}
\end{algorithm}

\begin{lemma}[Theorems 3.23, 3.24, \cite{dr14}]
\label{thm:Above_Threshold}
    $\AboTh$ is $(\epsilon,0)$-DP.
    Moreover, for $\gamma \in (0, 1)$, let $\alpha=\frac{8\Delta\log(\frac{2T}{\gamma})}{\epsilon}$ and $\calD \in \calX^n$. $\AboTh$ halts at time $k \in [T+1]$ such that with probability $\ge 1-\gamma$:
    \begin{itemize}
        \item $a_t =\bot$ and $q_t(\calD) \le \tau + \alpha$ for all $ t < k$.
        \item $a_k = \top$ and $q_k(\calD) \ge \tau - \alpha$ or $k = T+1$.
    \end{itemize} 
\end{lemma}

Finally, our developments in Section~\ref{sec:boosting} use the notions of R\'enyi DP (RDP) and central DP (CDP). We provide a self-contained summary of the definitions and properties satisfied by RDP and CDP here, but refer the reader to \cite{BunS16, Mironov17} for a more detailed overview.

\begin{definition}[RDP and CDP]
Let $\alpha \ge 1$, $\rho \ge 0$. We say that a randomized algorithm $\alg: \xset^n \to \Omega$ satisfies \emph{$(\alpha, \rho)$-R\'enyi differential privacy} (or, is \emph{$(\alpha, \rho)$-RDP}) if for all neighboring datasets $\calD, \calD' \in \xset^n$, 
\[D_\alpha(\alg(\calD) \| \alg(\calD')) \le \alpha\rho.\]
If this holds for all $\alpha \ge 1$, we say $\alg$ satisfies \emph{$\rho$-central differential privacy} (or, is \emph{$\rho$-CDP}).
\end{definition}

\begin{fact}[\cite{Mironov17}]\label{fact:rdp}
RDP and CDP satisfy the following properties.
\begin{enumerate}
    \item (Composition): If $\alg_1: \xset^n \to \Omega_1$ is $(\alpha, \rho_1)$-RDP and $\alg_2: \xset^n \times \Omega_1 \to \Omega_2$ is $(\alpha, \rho_2)$ for any fixed choice of input from $\Omega_1$, the composition of $\alg_2$ and $\alg_1$ is $(\alpha, \rho_1 + \rho_2)$-RDP.
    \item (RDP to DP): If $\alg$ is $(\alpha, \rho)$-RDP, it is also $(\alpha\rho + \frac 1 {\alpha - 1}\log \frac 1 \delta, \delta)$-DP for all $\delta \in (0, 1)$.
    \item (Gaussian mechanism): Let $\vv: \xset^n \to \R^k$ have sensitivity $\Delta$. Then for any $\sigma > 0$, drawing a sample from $\Nor(\vv(\calD), \sigma^2 \id_k)$ is $\frac{\Delta^2}{2\sigma^2}$-CDP.
\end{enumerate}
\end{fact}

\subsection{Geometric median}\label{ssec:median}

Throughout the rest of the paper, for a parameter $R > 0$, we fix a dataset $\calD \defeq \{\vx_i\}_{i \in [n]} \subset \ball^d(R)$, i.e., with domain $\xset \defeq \ball^d(R)$. Our goal is to approximate the \emph{geometric median} of $\calD$, i.e., 
\begin{equation}\label{eq:gm_def}
\begin{aligned}
\vxs(\calD) \defeq \arg\min_{\vx \in \R^d} \fcd(\vx),
\text{ where }
\fcd(\vx) \defeq \frac 1 n \sum_{i \in [n]} \norm{\vx - \vx_i} 
\end{aligned}
\end{equation}
is the average Euclidean distance to the dataset. Following e.g., \cite{cohen2016geometric, haghifam2024private}, we also define the quantile radii associated with our dataset $\calD$ centered at $\bvx \in \R^d$ by
\begin{equation}\label{eq:quantile_def}r^{(\tau)}(\calD; \bvx) \defeq \arg\min_{r \ge 0}\Brace{\sum_{i \in [n]} \ind_{\norm{\vx_i - \bvx} \le r} \ge \tau n}, \text{ for all } \tau \in [0, 1].\end{equation}
In other words, $r^{(\tau)}(\calD; \bvx)$ is the smallest radius $r \ge 0$ such that $\ball^d(\bvx, r)$ contains at least a $\tau$ fraction of the points in $\calD$. When $\bvx$ is unspecified, we always assume by default that $\bvx = \vxs(\calD)$.

In our utility analysis we will often suppress the dependence on $\calD$ in $\vxs, r^{(\tau)}$, etc., as the dataset of interest will not change. In the privacy analysis, we specify the dependence of these functions on the dataset explicitly when comparing algorithms run on neighboring datasets.

Finally, we include two helper results from prior work that are frequently used throughout.

\begin{lemma}\label{lem:opt_lb}
Let $\calD \defeq \{\vx_i\}_{i \in [n]} \subset \R^d$. Then, $\fcd\Par{\vxs} \ge (1 - \tau) r^{(\tau)}$ for all $\tau \in [0, 1]$.
\end{lemma}
\begin{proof}
This is immediate from the definition of $\fcd$ and nonnegativity of each summand $\norm{\cdot - \vx_i}$. 
\end{proof}

\begin{lemma}[Lemma 24, \cite{cohen2016geometric}]
\label{lm:property_of_gm}
Let $\calD \defeq \{\vx_i\}_{i \in [n]} \subset \R^d$ and let $S \subseteq [n]$ have $|S| < \frac{n}{2}$. Then, 
    \[
    \norm{\vxs - \vx} \le \left(\frac{2n - 2|S|}{n - 2|S|} \right) \max_{i \notin S} \|\vx_i - \vx\|,\text{ for all } \vx \in \R^d.
    \]
\end{lemma}


%% file: constant.tex
\section{Constant-Factor Approximation}\label{sec:constant}

In this section, we give our first main result: a fast algorithm for computing a constant-factor approximation to the geometric median. Our approach is to speed up several of the steps in the initial two phases of the \cite{haghifam2024private} algorithm via subsampling and techniques inspired by the FriendlyCore framework of \cite{tsfadia2022friendlycore}. Specifically, in Section~\ref{ssec:radius}, we first show how Algorithm~\ref{alg:mean_est_with_AT} can be sped up using subsampled scores, to estimate quantile radii up to constant factors in nearly-linear time, improving Section 2.1 of \cite{haghifam2024private}. In Section~\ref{ssec:center}, we then adapt a weighted variant of FriendlyCore to give a simple algorithm for approximate centerpoint computation, improving Section 2.2 of \cite{haghifam2024private} for large aspect ratios.

\subsection{Radius estimation}\label{ssec:radius}

In this section, we present and analyze our radius estimation algorithm.
\begin{algorithm}
\caption{$\FRF(\calD, r, R, \eps, \delta)$}
\label{alg:frf}
{\bf Input:} Dataset $\calD \in \ball^d(R)^n$, radius search bounds $0 < r \le R$, privacy bounds $(\eps, \delta) \in [0, 1]^2$\;
\begin{algorithmic}[1]
\STATE $T \gets \lceil\log_2(\frac R r)\rceil$\;
\STATE $k \gets 3\log(\frac {4T} \delta)$\;
\STATE $\tau \gets 0.775n$\;
\STATE $\hat{\tau} \gets \tau + \nu_{\mathsf{thresh}}$ for $\nu_{\thresh} \sim \Lap(\frac 6 \eps)$\;\label{line:noise_rad}
\FOR{$t \in [T]$}
\STATE $r_t \gets r \cdot 2^{t - 1}$\;
\FOR{$i \in [n]$}
\STATE $S^{(i)}_t \gets \Unif([n])^{\otimes k}$\;\label{line:subsample_rad}
\STATE $N^{(i)}_t \gets \frac n k \sum_{j \in S^{(i)}} \ind_{\norm{\vx_i - \vx_j} \le r_t}$\;
\ENDFOR
\STATE $q_t \gets \frac 1 n \sum_{i \in [n]} N^{(i)}_t$\;\label{line:subsample_stat}
\STATE $\nu_t \sim \Lap(\frac{12}{\epsilon})$\;\label{line:noise_rad_2}
\IF{$q_t +\nu_t \geq \hat{\tau}$}
\STATE {\bf Return:} $r_t$\;
\ENDIF
\ENDFOR
\STATE {\bf Return:} $R$\;
\end{algorithmic}
\end{algorithm}

Algorithm~\ref{alg:frf} is clearly an instance of Algorithm~\ref{alg:mean_est_with_AT} with $\Delta = 3$, where the queries are given on Line~\ref{line:subsample_stat}. However, one subtlety is that the queries in Algorithm~\ref{alg:frf} have random sensitivities depending on the subsampled sets on Line~\ref{line:subsample_rad}. Nonetheless, we show that Chernoff bounds control this sensitivity with high probability, which yields privacy upon applying Lemma~\ref{thm:Above_Threshold}.

\begin{lemma}\label{lem:privacy_finder}
Algorithm~\ref{alg:frf} is $(\eps, \delta)$-DP.
\end{lemma}
\begin{proof}
Fix neighboring datasets $\calD, \calD' \in \ball^d(R)^n$, and assume without loss that they differ in the $n^{\text{th}}$ entry. Observe that $\FRF$ (independently) uses randomness in two places: the random subsets in Line~\ref{line:subsample_rad}, and the Laplace noise added as in the original $\AboTh$ algorithm in Lines~\ref{line:noise_rad} and~\ref{line:noise_rad_2}. 

We next claim that in any iteration $t \in [T]$, as long as the number of copies of the index $n$ occurring in $\bigcup_{i \in [n - 1]} S_t^{(i)}$ is at most $2k$, then the sensitivity of the query $q_t$ is at most $3$. To see this, denoting by $q_t, q_t'$ the random queries when Algorithm~\ref{alg:frf} is run on $\calD$, $\calD'$ respectively, and similarly defining $\{N_t^{(i)}, (N_t^{(i)})'\}_{i \in [n]}$, we 
observe that the sensitivity is controlled as follows:
\begin{align*}
q_t - q'_t &\le \frac 1 n \sum_{i \in [n - 1]} \Par{N_t^{(i)} - \Par{N_t^{(i)}}'} + \frac n n \\
&\le \frac 1 n \cdot \frac n k \cdot \Par{\text{number of copies of } n \text{ occurring in } \bigcup_{i \in [n - 1]} S_t^{(i)}} + 1 \le 3.
\end{align*}
The first line holds because the $n^{\text{th}}$ (neighboring) point has $N_t^{(n)} \le n$ and $(N_t^{(n)})' \ge 0$; the second is because every  $\ind_{\norm{\vx_i - \vx_j} \le r_t}$ used in the computation of $N_t^{(i)}$ is coupled except when $j = n$ is sampled.

Now, let $\balg$ denote Algorithm~\ref{alg:frf}, and let $\alg$ denote a variant that conditions on the randomly-sampled $\bigcup_{i \in [n - 1]} S_t^{(i)}$ containing at most $2k$ copies of the index $n$, in all encountered iterations $t \in [T]$. By using Fact~\ref{fact:chernoff} (with $\mu \gets k \cdot \frac{n - 1}{n}$, $\eps \gets 1$), due to our choice of $k$, $\bigcup_{i \in [n - 1]} S_t^{(i)}$ contains at most $2k$ copies of $n$ except with probability $\frac \delta {4T}$, so by a union bound, the total variation distance between $\balg$ and $\alg$ is at most $\frac \delta 4$. Moreover, $\alg$ is $(\eps, 0)$-DP by using Lemma~\ref{thm:Above_Threshold}. Thus, $\balg$ is $(\eps, \delta)$-DP using Lemma~\ref{lem:dp_close}.
\end{proof}

We are now ready to prove a utility and runtime guarantee on Algorithm~\ref{alg:frf}.

\begin{lemma}\label{lem:frf_utility}
Algorithm~\ref{alg:frf} runs in time $O(nd\log(\frac R r)\log(\log(\frac R r)\frac 1 \delta))$.
Moreover, if $r \le 4r^{(0.9)}$ and $n \ge \frac{2400}{\eps}\log(\frac{4T}{\delta})$, with probability $\ge 1 - \delta$, Algorithm~\ref{alg:frf} outputs $\hr$ satisfying
$\frac 1 4 r^{(0.75)} \le \hr \le 4r^{(0.9)}$.
\end{lemma}
\begin{proof}
The first claim is immediate. 
To see the second, for all $t \in [T]$ denote the ``ideal'' query by:
\[q^\star_t \defeq \frac 1 n \sum_{i \in [n]} \sum_{j \in [n]} \ind_{\norm{\vx_i - \vx_j} \le r_t},\]
and recall $\E q_t = q^\star_t$. Our first claim is that with probability $\ge 1 - \frac \delta {2}$, the following guarantees hold for all iterations $t \in [T]$ that Algorithm~\ref{alg:frf} completes:
\begin{equation}\label{eq:count_preserved}
q^\star_t > 0.8n \implies q_t > 0.79n, \text{ and } q^\star_t < 0.75n \implies q_t < 0.76n.
\end{equation}
To see the first part of \eqref{eq:count_preserved}, we can view $nq_t$ as a random sum of Bernoulli variables with mean $nq_t^\star > 0.8n^2 \ge 20000\log(\frac{\delta}{4T})$, so Fact~\ref{fact:chernoff} with $\eps \gets \frac 1 {80}$ yields the claim in iteration $t$ with probability $\ge 1 - \frac \delta {2T}$. Similarly, the second part of \eqref{eq:count_preserved} follows by using Fact~\ref{fact:chernoff} with $\mu < 0.7n^2$ and $(1 + \eps)\mu \gets 0.71n^2$, because
\begin{equation}\label{eq:large_eps}\exp\Par{-\frac{\eps^2 \mu}{2 + \eps}} \le \exp\Par{-\frac{\eps^2 \mu}{151\eps}} = \exp\Par{-\frac{\eps\mu}{151}} \le \exp\Par{-\frac{n^2}{15100}} \le \frac \delta {2T}\end{equation}
for the relevant range of $n$ and $\eps \ge \frac 1 {75}$, $\eps \mu \ge \frac {n^2}{100}$. We thus obtain \eqref{eq:count_preserved} after a union bound over all $t \in [T]$.

Now, suppose that $t \in [T]$ is the first index where $q_t + \nu_t \ge \htau$, so that $\hr = r_t$ and $\half \hr = r_{t - 1}$, where we let $r_0 \defeq \frac r 2$. If no such query passes, then we set $t = T + 1$ by default. Then by the utility guarantees of Lemma~\ref{lem:frf_utility}, we have that with probability $\ge 1 - \frac \delta 2$,
\begin{align*}
q_t \ge 0.76n,\; q_{t - 1} \le 0.79n,
\end{align*}
since $\frac n {100} \ge \alpha = \frac{24}{\eps}\log(\frac{4T}{\delta})$ . By taking the contrapositive of \eqref{eq:count_preserved}, we can conclude that with probability $\ge 1 - \delta$, we have $q^\star_t \ge 0.75n$ and $q^\star_{t - 1} \le 0.8n$. Condition on this event for the rest of the proof.

Because $q_t^\star \ge 0.75n$, there is clearly some $\vx_i \in \calD$ such that $|\calD \cap \ball^d(\vx_i, \hr)| \ge 0.75n$, as this is the average number of dataset elements in a radius-$\hr$ ball centered at a random $\vx \in \calD$. Now applying Lemma~\ref{lm:property_of_gm} with $S$ set to the indices of $\calD \setminus \ball^d(\vx_i, \hr)$, so that $|S| \le 0.25n$, gives
\[\norm{\vxs - \vx_i} \le 3\hr \implies \norm{\vxs - \vx_j} \le 4\hr \text{ for all } \vx_j \in \calD \cap \ball^d(\vx_i, \hr).\]
This implies $4\hr \ge r^{(0.75)}$ as claimed. Further, because $q_{t - 1}^\star \le 0.8n$, we claim $\frac \hr 2 > 2r^{(0.9)}$ cannot hold. Assume for contradiction that this happened, and let $S \defeq \calD \cap \ball^d(\vxs, r^{(0.9)} )$. By the triangle inequality, for all of the $0.9n$ choices of $\vx_i \in S$, we have that 
\[\sum_{j \in [n]} \ind_{\norm{\vx_i - \vx_j} \le \frac{\hr}{2}} \ge 0.9n,\]
which implies that $q^\star_{t - 1} \ge 0.81n$, a contradiction. Thus, we obtain $\hr \le 4r^{(0.9)}$ as well.

We remark that all of this logic handles the case where $2 \le t \le T$ is the iteration where Algorithm~\ref{alg:frf} returns. However, it is straightforward to check that the conclusion holds when $t = 1$ (i.e., $\hr = r$) because we assumed $r \le 4r^{(0.9)}$, and the lower bound logic on $\hr$ is the same as before. Similarly, if $\hr = R$, then the upper bound logic on $\hr$ is the same as before, and $2R \ge r^{(1)} \ge r^{(0.75)} \ge \frac 2 4 r^{(0.75)}$ is clear.
\end{proof}

In summary, Lemmas~\ref{lem:privacy_finder} and~\ref{lem:frf_utility} show that we can privately estimate $\hr$ satisfying $\frac 1 4 r^{(0.75)} \le \hr \le 4r^{(0.9)}$ in nearly-linear time. The upper bound implies (with Lemma~\ref{lem:opt_lb}) that $\hr = O(f_{\calD}(\vxs))$; on the other hand, the lower bound will be critically used in our centerpoint estimation procedure in Section~\ref{ssec:center}.

\subsection{Centerpoint estimation}\label{ssec:center}

In this section, we combine the subsampling strategies used in Section~\ref{ssec:radius} with a simplification of the FriendlyCore framework \cite{tsfadia2022friendlycore} to obtain a private estimate of an approximate centerpoint. 

\begin{algorithm}
\caption{$\FC(\calD, \hr, \eps, \delta)$}
\label{alg:fc}
{\bf Input:} Dataset $\calD \in \ball^d(R)^n$, radius $\hr \in \R_{>0}$, privacy bounds $(\eps, \delta) \in [0, 1]^2$\;
\begin{algorithmic}[1]
\STATE $k \gets 600\log(\frac{18n}{\delta})$\;
\FOR{$i \in [n]$}
\STATE $S_i\gets \Unif([n])^{\otimes k}$\label{line:subsamp_center}
\STATE $f_i \gets \sum_{j \in S^{(i)}} \ind_{\norm{\vx_i - \vx_j} \le 2\hr}$\;
\STATE $p_i \gets \min(\max(0, \frac{f_i - 0.5k}{0.25k}), 1)$\;
\ENDFOR
\STATE $Z \gets \sum_{i \in [n]} p_i$\;
\STATE $\xi \sim \BL(\frac{24}{\eps}, \frac{24}{\eps}\log(\frac{24}{\delta}))$\;
\IF{$Z + \xi - \frac{24}{\eps}\log(\frac{24}{\delta}) \le 0.55n$}{\label{line:Zsmall}
\STATE {\bf Return:} $\0_d$\;\label{line:Zsmall_return}
}
\ENDIF
\STATE $\bvx \gets \frac 1 {Z} \sum_{i \in [n]} p_i \vx_i$\; \label{line:xavg_compute}
\STATE $\vxi \sim \Nor(\0_d, \sigma^2 \id_d)$, for $\sigma \gets \frac{1600\hr}{n\eps} \sqrt{\log(\frac {12} \delta)}$\;
\STATE {\bf Return:} $\bvx + \vxi$\;\label{line:return_center}
\end{algorithmic}
\end{algorithm}

To briefly explain, Algorithm~\ref{alg:fc} outputs a noisy weighted average of the dataset. The weights $\{p_i\}_{i \in [n]}$ linearly interpolate estimated scores $f_i \in [0.5k, 0.75k]$ into the range $[0, 1]$, sending $f_i \ge 0.75k$ to $1$, and $f_i \le 0.5k$ to $0$.
We first make some basic observations about the points that contribute positively to the weighted combination $\bvx$, based on binomial concentration.

\begin{lemma}\label{lem:center_utility}
Assume that $\hr \ge r^{(0.75)}$ in the context of Algorithm~\ref{alg:fc}.
With probability $\ge 1 - \frac \delta {18}$, every $i \in [n]$ that is assigned $p_i > 0$ in Algorithm~\ref{alg:fc} satisfies $\norm{\vx_i - \vxs} \le 3\hr$, and $Z \ge 0.6n$.
\end{lemma}
\begin{proof}
Our proof is analogous to Lemma~\ref{lem:frf_utility}, where for all $i \in [n]$ we define the ``ideal score''
\[f^\star_i \defeq \frac k n \sum_{j \in [n]} \ind_{\norm{\vx_i - \vx_j} \le 2\hr},\]
such that $\E f_i = f^\star_i$. We first claim that with probability $\ge 1 - \frac \delta {18}$, the following hold for all $i \in [n]$:
\begin{equation}\label{eq:binom_conc_fc}f^\star_i \ge 0.75k \implies f_i \ge 0.7k, \text{ and } f^\star_i \le 0.45k \implies f_i \le 0.5k. \end{equation}
The first claim above is immediate from our choice of $k$ and Fact~\ref{fact:chernoff} (with failure probability $\le \frac \delta {18n}$ for each $i \in [n]$); the second follows (with the same failure probability) similarly to \eqref{eq:large_eps}, i.e.,
\[\exp\Par{-\frac{\eps^2 \mu}{2+\eps}} \le \exp\Par{-\frac{\eps\mu}{20}} \le \exp\Par{-\frac{k}{400}} \le \frac \delta {18n},\]
in our application, with $\eps \ge \frac 1 9$ and $\eps\mu \ge \frac k {20}$. Thus a union bound proves \eqref{eq:binom_conc_fc}. 

To obtain the first claim, observe that any $i \in [n]$ with $\norm{\vx_i - \vxs} > 3\hr$ must have that $\ball(\vx_i, 2\hr)$ does not intersect $\ball(\vxs, \hr)$. However, $\ball(\vxs, \hr)$ contains $0.75n$ points in $\calD$ by assumption, so $f^\star_i \le 0.25k$ and thus as long as the implication \eqref{eq:binom_conc_fc} holds, then $p_i = 0$ as desired. For the second claim, any $\vx_i$ satisfying $\norm{\vx_i - \vxs} \le r^{(0.75)}$ has $|\ball(\vx_i, 2\hr) \cap \calD| \ge 0.75n$, so that $f^\star_i \ge 0.7k$. Thus, every such $\vx_i$ has $p_i \ge 0.8$ as long as \eqref{eq:binom_conc_fc} holds, so the total contribution made by the $\ge 0.75n$ such $\vx_i$ to $Z$ is at least $0.6n$.
\end{proof}

We next observe that whenever the algorithm does not return on Line~\ref{line:Zsmall_return}, all surviving points (i.e., with $p_i > 0$) must lie in a ball of diameter $O(\hr)$, under a high-probability event over our subsampled scores. Importantly, this holds independently of any assumption on $\hr$ (e.g., we do not require $\hr \ge r^{(0.75)}$).

\begin{lemma}\label{lem:survive_ball}
Suppose that it is the case that in the context of Algorithm~\ref{alg:fc}, we have
\begin{equation}\label{eq:fc_conc_imply}
f^\star_i \le 0.45k \implies f_i \le 0.5k, \text{ and } f^\star_i \le 0.55k \implies f_i \le 0.6k,
\end{equation}
for all $i \in [n]$. If $Z > 0.55n$, there exists some $\vx \in \R^d$ such that $\{\vx_i \mid p_i > 0\}_{i \in [n]} \subseteq \ball(\vx, 4\hr)$. Moreover, the event \eqref{eq:fc_conc_imply} occurs with probability $\ge 1 - \frac \delta 9$.
\end{lemma}
\begin{proof}
The first statement in \eqref{eq:fc_conc_imply} was proven in \eqref{eq:binom_conc_fc} to hold with probability $\frac \delta {18}$, and the second statement's proof is identical to the first half up to changing constants, so we omit it. Conditioned on this event, every $\vx_i$ with $p_i > 0$ has $f^\star_i > 0.45k$. Moreover, because $Z > 0.55n$, there exists some $j \in [n]$ (i.e., with the maximum value of $p_j$) such that $p_j > 0.55$, which implies $f_j > 0.6k$ and thus $f^\star_j > 0.55k$.

So, we have shown that $\ball(\vx_j, 2\hr)$ contains more than $0.55n$ points in $\calD$, and every surviving $i \in [n]$ (i.e., with positive $p_i$) contains more than $0.45n$ points in $\calD$. Thus, $\ball(\vx_j, 2\hr)$ and $\ball(\vx_i, 2\hr)$ intersect, and in particular, $\ball(\vx_j, 4\hr)$ contains every surviving point by the triangle inequality.
\end{proof}

We are now ready to prove a privacy bound on Algorithm~\ref{alg:fc}.

\begin{lemma}\label{lem:fc_priv}
If $n \ge 20$, Algorithm~\ref{alg:fc} is $(\eps, \delta)$-DP.
\end{lemma}
\begin{proof}
Fix neighboring datasets $\calD, \calD' \in \ball^d(R)^n$, and assume without loss that they differ in the $n^{\text{th}}$ entry $\vx_n \neq \vx'_n$. We will define $\alg$, an alternate variant of Algorithm~\ref{alg:fc}, which we denote $\balg$, where we condition on the following two events occurring. First, the index $n$ should occur at most $2k$ times in $\bigcup_{i \in [n - 1]} S_i$. Second, the implications \eqref{eq:fc_conc_imply} must hold. It is clear that the first described event occurs with probability $\ge 1 - \frac \delta {18}$ by using Fact~\ref{fact:chernoff} with our choice of $k$, and we proved in Lemma~\ref{lem:survive_ball} that the second described event occurs with probability $\ge 1 - \frac \delta 9$. Thus, the total variation distance between $\alg$ and $\balg$ is at most $\frac \delta 6$. We will prove that $\alg$ is $(\eps, \frac \delta 3)$-DP, from which Lemma~\ref{lem:dp_close} gives that $\balg$ is $(\eps, \delta)$-DP.

We begin by showing that according to $\alg$, the statistic $Z + \xi - \frac {24} \eps \log(\frac{24}{\delta})$ satisfies $(\frac \eps 2, \frac \delta 6)$-DP. To do so, we will prove that $Z$ has sensitivity $\le 12$, and then apply Fact~\ref{fact:bl_mech}. 
Recall that by assumption, when $\alg$ is run the number of times $n$ appears in $\bigcup_{i \in [n]} S_i$ is at most $2k$. Thus, for coupled values of $Z, Z'$ corresponding to $\calD, \calD'$, where the coupling is over the random indices selected on Line~\ref{line:subsamp_center}, 
\[Z - Z' \le \frac{1}{0.25k}\Par{k - 0} + \frac 1 {0.25k} \Par{\text{number of copies of } n \text{ occurring in } \bigcup_{i \in [n - 1]} S_i} \le 12.\]
In fact, we note that the following stronger unsigned bound holds:
\begin{equation}\label{eq:abs_p_bound}
\begin{aligned}
\sum_{i \in [n]} |p_i - p'_i| &\le 4 + \sum_{i \in [n - 1]} |p_i - p'_i| \\
&\le 4 + \frac 1 {0.25k} \Par{\text{number of copies of } n \text{ occurring in } \bigcup_{i \in [n - 1]} S_i} \le 12,
\end{aligned}
\end{equation}
because the clipping to the interval $[0, 1]$ in the definitions of $p_i$, $p'_i$ can only improve $|p_i - p'_i|$, and the distance between the corresponding $f_i, f'_i$ is at most the number of copies of $n$ occurring in them.

Now, it remains to bound the privacy loss of the rest of $\alg$, depending on whether Line~\ref{line:Zsmall} passes. If the algorithm terminates on Line~\ref{line:Zsmall_return}, then there is no additional privacy loss. 

Otherwise, suppose we enter the branch starting on Line~\ref{line:xavg_compute}. 
Our next step is to bound the sensitivity of $\bvx$. Observe that whenever this branch is entered, we necessarily have $Z > 0.55n$ (and similarly, $Z' > 0.55n$), because $Z + \xi - \frac{24}{\eps}\log(\frac{24}{\delta}) \le Z$ deterministically. Thus, Lemma~\ref{lem:survive_ball} guarantees that in $\alg$, all $\vx_i \in \calD$ with $p_i > 0$ are contained in a ball of radius $4\hr$, and similarly all surviving elements in $\calD'$ are contained in a ball of radius $4\hr$. However, there are at least $0.55n$ surviving elements of both $\calD$ and $\calD'$, and in particular, for the given range of $n$ there are at least two common surviving elements (one of which must be shared). A ball of radius $8\hr$ around this element, which we denote $\hvx$ in the rest of the proof, contains all surviving elements in $\calD$ (according to $\{p_i\}_{i \in [n]}$) and in $\calD'$ (according to $\{p'_i\}_{i \in [n]}$).

Now we wish to bound $\bvx - \bvx'$, where $\bvx' \defeq \frac 1 {Z'} \sum_{i \in [n]} p'_i \vx'_i$. We have shown that in $\alg$, $|Z - Z'| \le 12$ and $\min(Z, Z') \ge 0.55n$. For convenience, define $\vy_i \defeq \vx_i - \hvx$ for all $i \in [n]$, and similarly define $\vy'_i$. Recalling that all surviving elements of $\calD \cup \calD'$ are contained in $\ball(\hvx, 8\hr)$,
\begin{align*}
\norm{\bvx - \bvx'} &= \norm{\frac 1 Z \sum_{i \in [n]} p_i \vy_i - \frac 1 {Z'} \sum_{i \in [n]} p'_i \vy'_i} \\
&\le \Abs{\frac 1 Z - \frac 1 {Z'}} \norm{\sum_{i \in [n - 1]} p_i \vy_i} + \frac 1 {Z'} \norm{\sum_{i \in [n - 1]} (p_i - p'_i) \vy_i} + \frac{p_n}{Z} \norm{\vy_n} + \frac{p'_n}{Z'} \norm{\vy'_n} \\
&\le \Abs{\frac {Z' - Z} {Z'}} \norm{\frac 1 Z\sum_{i \in [n - 1]} p_i \vy_i} + \frac{8\hr}{0.55n} \sum_{i \in [n - 1]} |p_i - p'_i| + \frac{16\hr}{0.55n} \\
&\le \frac{96\hr}{0.55n} + \frac{104\hr}{0.55n} + \frac{16\hr}{0.55n} \le \frac{400\hr}{n}.
\end{align*}
The first line shifted both $\bvx$ and $\bvx'$ by $\hvx$, and the second line applied the triangle inequality. The third line applied the triangle inequality to the middle term, and bounded the contribution of $\vy_n$ by using that $\norm{\vy_n} \le 8\hr$ if $p_n > 0$; a similar bound applies to $\vy'_n$.
In the fourth line, we used the triangle inequality on the first term, as well as that $\sum_{i \in [n - 1]} |p_i - p'_i| \le 1 + 12$ by using \eqref{eq:abs_p_bound} and accounting for the $n^{\text{th}}$ point separately.
Thus, $\bvx$ has sensitivity $\frac{400\hr}{n}$ in $\alg$. Fact~\ref{fact:gaussian_mech} now guarantees that Line~\ref{line:return_center} is also $(\frac \eps 2, \frac \delta 6)$-DP.
\end{proof}

We now combine our developments to give our constant-factor approximation to the geometric median.

\begin{theorem}\label{thm:constant_factor}
Let $\calD = \{\vx_i\}_{i \in [n]} \subset \ball^d(R)$ for $R > 0$, let $0 < r \le 4r^{(0.9)}(\calD)$, and let $(\eps, \delta) \in [0, 1]^2$. Suppose that 
\[n \ge C \cdot \Par{\frac{\sqrt{d}\log(\frac 1 \delta)}{\eps} + \frac{\log\Par{\frac {\log(\frac R r)} {\delta} }}{\eps} },\]
for a sufficiently large constant $C$. There is an $(\eps, \delta)$-DP algorithm (Algorithm~\ref{alg:fc} using Algorithm~\ref{alg:frf} to compute the parameter $\hr$) that returns $(\hvx, \hr)$ such that with probability $\ge 1 - \delta$, following notation \eqref{eq:gm_def}, 
\begin{equation}\label{eq:func_lip}
f_{\calD}(\hvx) \le (40C' + 1) f_{\calD}(\vxs(\calD)),\; \hr \le 4r^{(0.9)},
\end{equation}
for a universal constant $C'$.
Moreover, $\norm{\vxs(\calD) - \hvx} \le C' \hr$. The algorithm runs in time
\[O\Par{nd\log\Par{\frac R r}\log\Par{\frac{n\log(\frac R r)} \delta}}.\]
\end{theorem}
\begin{proof}
Regarding the utility bound, we will only establish that $\norm{\vxs(\calD) - \hvx} \le C' \hr$, which also gives \eqref{eq:func_lip} upon observing that $f_{\calD}$ is $1$-Lipschitz, and $f_{\calD}(\vxs(\calD)) \ge 0.1r^{(0.9)}(\calD)$, due to Lemma~\ref{lem:opt_lb}.

We first run Algorithm~\ref{alg:frf} with parameters $(\eps, \delta) \gets (\frac \eps 2, \frac \delta 2)$, which gives for large enough $C$ (via Lemma~\ref{lem:frf_utility})
\begin{equation}\label{eq:phase1_utility}\frac 1 4 r^{(0.75)} \le \hr \le 4r^{(0.9)},\end{equation}
with probability $\ge 1 - \frac \delta 2$. Next, we run Algorithm~\ref{alg:fc} with this value of $\hr$, and parameters $(\eps, \delta) \gets (\frac \eps 2, \frac \delta 2)$. The privacy of composing these two algorithms now follows from Lemmas~\ref{lem:privacy_finder} and~\ref{lem:fc_priv}, and the runtime follows from Lemma~\ref{lem:frf_utility}, because Algorithm~\ref{alg:frf}'s runtime does not dominate upon inspection.

It remains to argue about the utility, i.e., that $\norm{\vxs(\calD) - \hvx} \le C' \hr $. Conditioned on \eqref{eq:phase1_utility} holding, Lemma~\ref{lem:frf_utility} guarantees that with probability $\ge 1 - \frac \delta 4$, we have that $\norm{\bvx - \vxs(\calD)} \le 3\hr$, as a positively-weighted average of points in $\ball^d(\vxs(\calD), 3\hr)$. Finally, for the given value of $\sigma$ in Algorithm~\ref{alg:frf}, standard Gaussian concentration bounds imply that with probability $\ge 1 - \frac \delta 4$,
\[\norm{\bvx - \hvx} = \norm{\vxi} \le 3\sigma \sqrt{d\log\Par{\frac 4 \delta}} = O\Par{\hr \cdot \frac{\sqrt d \log(\frac 1 \delta)}{n\eps}} = O(\hr) .\]
Thus, $\norm{\vxs(\calD) - \hvx} \le C' \hr$ holds for an appropriate $C'$, except with probability $\delta$.
\end{proof}

We remark that Theorem~\ref{thm:constant_factor} actually comes with the slightly stronger guarantee that we obtain the optimal value for the geometric median objective $f_{\calD}$, up to an additive error scaling as $O(r^{(0.9)})$. In general, while $r^{(0.9)} = O(f_{\calD}(\vxs(\calD)))$ is always true by Lemma~\ref{lem:opt_lb}, it is possible that $r^{(0.9)} \ll f_{\calD}(\vxs(\calD))$ if a small fraction of outlier points contributes significantly to the objective. We also note that for datasets where we have more a priori information on the number of outliers we expect to see, we can adjust the quantile $0.9$ in Theorem~\ref{thm:constant_factor} to be any quantile $> 0.5$ by appropriately adjusting constants.

%% file: boosting.tex
\section{Boosting Approximations via Stable DP-SGD}\label{sec:boosting}

In this section, we give a DP algorithm that efficiently minimizes the geometric median objective \eqref{eq:gm_def}
over a domain $\ball^d(\bvx, \hr)$, given a dataset $\calD \defeq \{\vx_i\}_{i \in [n]}$. In our final application to the geometric median problem, the optimization domain (i.e., the parameters $\bvx \in \R^d$ and $\hr \in \R_{\ge 0}$) will be privately estimated using Theorem~\ref{thm:constant_factor}, such that with high probability $\hr = O(f_{\calD}(\vxs(\calD)))$ and $\norm{\vx - \vxs(\calD)} \le \hr$. In the meantime, we treat the domain $\ball^d(\bvx, \hr)$ as a public input here.

Our strategy is to use a localization framework given by \cite{FeldmanKT20}, which gives a query-efficient reduction from private DP-SGD to stable DP-SGD executed in phases. Specifically, observe that outputting\footnote{By default, if $\vx = \vx_i$, we let \eqref{eq:sample_grad} evaluate to $\0_d$, which is a valid subgradient by first-order optimality.}
\begin{equation}\label{eq:sample_grad}\frac{\vz - \vx_i}{\norm{\vz - \vx_i}} = \nabla\norm{\cdot - \vx_i}\Par{\vz}\end{equation}
for a uniformly random $i \in [n]$ is unbiased for a subgradient of $f_{\calD}(\vz)$. This leads us to define the following Algorithm~\ref{alg:phased_sgd} patterned off the \cite{FeldmanKT20} framework, whose privacy is analyzed in Section~\ref{ssec:boost_priv} using a custom stability argument, and whose utility is analyzed in Section~\ref{ssec:boost_util}.

\begin{algorithm}[h]
\caption{$\StableDPSGD(\calD, \bvx, \hr, \rho, \delta, \eta, T)$}
\label{alg:phased_sgd}
\textbf{Input:} Dataset $\calD = \{\vx_i\}_{i \in [n]} \subset \R^d$, domain parameters $(\bvx, \hr) \in \R^d \times \R_{\ge 0}$, privacy bound $\rho > 0$, failure probability $\delta \in (0, 1)$, step size $\eta > 0$, step count $T = 2^K - 1 \ge n$ for $K \in \N$
\begin{algorithmic}[1]
\STATE $m \gets 3(\frac {T} n + \log(\frac{8}{\delta}))$\;
\FOR{$k \in [K]$}\label{line:phase_start}
    \STATE $(T^{(k)}, \eta^{(k)}, \sigma^{(k)}) \gets (2^{-k}(T + 1), 4^{-k}\eta, 3^{-k} \frac{(2m + 1)\eta}{\sqrt{\rho}})$\;
    \IF{$k = 1$}
        \STATE $\vz^{(k)}_0 \gets \bvx$\;
        \STATE $\set^{(k)} \gets \ball^d(\vz^{(k)}_0, \hr)$\;
    \ELSE
        \STATE $\vz^{(k)}_0 \gets \hvx^{(k - 1)}$
        \STATE $\set^{(k)} \gets \ball^d(\vz_0^{(k)}, 2\sigma^{(k)} \sqrt{d\log(\frac{4K}{\delta})})$
    \ENDIF 
    \FOR{$0 \le t < T^{(k)}$}
        \STATE $\vg^{(k)}_t \gets \frac{\vz_t^{(k)} - \vx_i}{\norms{\vz_t^{(k)} - \vx_i}}$ for $i \sim \Unif([n])$\;\label{line:sample_index}
        \STATE $\vz_{t + 1}^{(k)} \gets \proj_{\set^{(k)}}(\vz_t^{(k)} - \eta^{(k)} \vg_t^{(k)})$\;
    \ENDFOR
    \STATE $\bvx^{(k)} \gets \frac 1 {T^{(k)}} \sum_{0 \le t < T^{(k)}} \vz_t^{(k)}$\;\label{line:avg_start}
    \STATE $\vxi^{(k)} \sim \Nor(\0_d, (\sigma^{(k)})^2 \id_d)$\;
    \STATE $\hvx^{(k)} \gets \bvx^{(k)} + \vxi^{(k)}$\;\label{line:avg_end}
\ENDFOR\label{line:phase_end}
\STATE {\bf Return:} $\hvx^{(K)}$\;\label{line:return_center}
\end{algorithmic}
\end{algorithm}

Algorithm~\ref{alg:phased_sgd} proceeds in $K \approx \log(T)$ phases. In each phase (loop of Lines~\ref{line:phase_start} to~\ref{line:phase_end}) other than $k = 1$, we define a domain $\set^{(k)}$ centered at the output of the previous phase with geometrically shrinking radius $\propto \sigma^{(k)}$; the domain for phase $k = 1$ is simply $\ball^d(\bvx, \hr)$. After this, we take $T^{(k)}$ steps of SGD over $\set^{(k)}$ with step size $\eta^{(k)}$, and output a noised variant of the average iterate in Lines~\ref{line:avg_start} to~\ref{line:avg_end}.

\begin{remark}\label{rem:practical}
Several steps in Algorithm~\ref{alg:phased_sgd} are used only in the worst-case utility proof, and do not affect privacy. Practical optimizations can be made while preserving privacy guarantees, e.g., removing projections onto the changing domains $\set^{(k)}$ rather than $\set^{(1)}$, which is not used in the privacy proof.

One optimization we found useful in our experiments (described in Section~\ref{sec:experiments}) is replacing the random sampling on Line~\ref{line:sample_index} with deterministic passes through the dataset in a fixed order. By doing so, we know the total number of accesses of any single element is $\le m \defeq \lceil \frac T n \rceil$ (rather than the high-probability estimate in Lemma~\ref{lem:m_small} for the randomized variant in Algorithm~\ref{alg:phased_sgd}). This lets us tighten the noise level $\sigma^{(k)}$ by a fairly significant constant factor, resulting in improved empirical performance.
\end{remark}

\subsection{Privacy of Algorithm~\ref{alg:phased_sgd}}\label{ssec:boost_priv}

In this section, we show that Algorithm~\ref{alg:phased_sgd} satisfies $(\eps, \delta)$-DP for an appropriate choice of $\rho$. When the sample functions of interest are smooth (i.e., have bounded second derivative), \cite{FeldmanKT20} gives a proof based on the \emph{contractivity} of iterates. This is based on the observation that gradient descent steps with respect to a smooth function are contractive for an appropriate step size (see e.g., Proposition 2.10, \cite{FeldmanKT20}). Unfortunately, our sample functions are of the form $\norm{\cdot - \vx_i}$, which are not even differentiable, let alone smooth. Nonetheless, we inductively prove approximate contractivity of Algorithm~\ref{alg:phased_sgd}'s iterates by opening up the analysis and using the structure of the geometric median objective.

Throughout, we fix neighboring $\calD, \calD' \in (\R^d)^n$, and assume without loss of generality they differ in the $n^{\text{th}}$ entry. 
To simplify notation, we let $\calI \in [n]^T$ denote the multiset of $T$ indices sampled in Line~\ref{line:sample_index}, across all phases. 
We prove DP of Algorithm~\ref{alg:phased_sgd} via appealing to Lemma~\ref{lem:dp_close}, where we let $\balg$ denote Algorithm~\ref{alg:phased_sgd}, and we let $\alg$ denote a variant of Algorithm~\ref{alg:phased_sgd} conditioned on $\calI$ containing at most $m \defeq 3(\frac T n + \log(\frac 8 \delta))$ copies of $n$. We first bound the total variation distance between $\alg$ and $\balg$ using Fact~\ref{fact:chernoff}.

\begin{lemma}\label{lem:m_small}
With probability $\ge 1 - \frac \delta 8$, Algorithm~\ref{alg:phased_sgd} yields $\calI$ containing $\le m$ copies of $n$.
\end{lemma}
\begin{proof}
In expectation, we have $\frac T n \ge 1$ copies, so the result follows from Fact~\ref{fact:chernoff} and our choice of $m$.
\end{proof}

We will show that $\alg$ is $(\eps, \frac \delta 2)$-DP, upon which Lemmas~\ref{lem:dp_close} and~\ref{lem:m_small} imply that $\balg$ (Algorithm~\ref{alg:phased_sgd}) is $(\eps, \delta)$-DP. To do so, we control the sensitivity of each iterate $\vz_t^{(k)}$, using the following two helper facts.

\begin{fact}[\cite{Rockafellar76}]\label{fact:project_contract}
Let $\set \subset \R^d$ be a compact, convex set. Then for any $\vx, \vy \in \R^d$, we have
\[\norm{\proj_\set(\vx) - \proj_\set(\vy)} \le \norm{\vx - \vy}.\]
\end{fact}

\begin{lemma}\label{lem:same_sgd_step}
For any unit vectors $\vu, \vv \in \R^d$, and $a, b > 0$, let $\vx = a\vu$ and $\vy = b\vv$. Then, letting $\vx' \gets (a - \eta)\vu$ and $\vy' \gets (b - \eta)\vv$, we have $\norm{\vx' - \vy'} \le \max(\norm{\vx - \vy}, 3\eta)$.
\end{lemma}
\begin{proof}
We claim that
\begin{equation}\label{eq:contract_condition}\norm{\vx' - \vy'} \le \norm{\vx - \vy} \iff a + b \ge \eta,\end{equation}
from which the proof follows from observing that if $a + b \le \eta$, then we can trivially bound
$\norm{\vx' - \vy'} \le \norm{\vx - \vy} + 2\eta \le a + b + 2\eta \le 3\eta$.
Indeed, \eqref{eq:contract_condition} follows from a direct expansion:
\begin{align*}
\norm{\vx - \vy}^2 - \norm{\vx' - \vy'}^2 &= a^2 + b^2 - 2ab\inprod{\vu}{\vv} - (a - \eta)^2 - (b - \eta)^2 + 2(a - \eta)(b - \eta)\inprod{\vu}{\vv} \\
&= 2\eta(a + b) - 2\eta^2 - 2\eta(a + b)\inprod{\vu}{\vv} + 2\eta^2\inprod{\vu}{\vv} \\
&= 2\eta(a + b - \eta)(1 - \inprod{\vu}{\vv}).
\end{align*}
Thus, for $\eta \ge 0$ and $\inprod{\vu}{\vv} \ge 0$, we conclude that \eqref{eq:contract_condition} holds.
\end{proof}

\begin{corollary}\label{cor:phase_stable}
For any phase $k \in [K]$ in Algorithm~\ref{alg:phased_sgd}, condition on the value of $\vz_0^{(k)}$, and assume that $\calI$ contains at most $m$ copies of $n$. Then for any $0 \le t < T_k$, the sensitivity of $\vz_t^{(k)}$ is $\le (2m + 1)\eta^{(k)}$.
\end{corollary}
\begin{proof}
Throughout this proof only, we drop the iteration $k$ from superscripts for notational simplicity, so we let $T \defeq T^{(k)}$ and $\eta \defeq \eta^{(k)}$, referring to the relevant iterates as $\{\vz_t\}_{0 \le t < T} \defeq \{\vz_t^{(k)}\}_{0 \le t < T^{(k)}}$. We also refer to the index $i$ selected on Line~\ref{line:sample_index} in iteration $0 \le t < T$ by $i_t$.

Fix two copies of the $k^{\text{th}}$ phase of Algorithm~\ref{alg:phased_sgd}, both initialized at $\vz_0$, but using neighboring datasets $\calD$, $\calD'$ differing in the $n^{\text{th}}$ entry. Also, fix a realization of $\{i_t\}_{0 \le t < T}$, such that $i_t = n$ at most $m$ choices of $t$ (note that $m$ is actually a bound on how many times $i_t = n$ across all phases, so it certainly bounds the occurrence count in a single phase). Conditioned on this realization, Algorithm~\ref{alg:phased_sgd} is now a deterministic mapping from $\vz_0$ to the iterates $\{\vz_t\}_{0 \le t < T}$, depending on the dataset used.

Denote the iterates given by the dataset $\calD$ by $\{\vz_t\}_{0 \le t < T}$ and the iterates given by $\calD'$ by $\{\vz'_t\}_{0 \le t < T}$, so that $\vz_0 = \vz'_0$ by assumption. Also, let $\Phi_t \defeq \norm{\vz_t - \vz'_t}$ for all $0 \le t < T$. We claim that for all $0 \le t < T$, 
\begin{equation}\label{eq:sensitivity}\Phi_t \le \max(2m_t + 1, 3)\eta, \text{ where } m_t \defeq \sum_{0 \le s < t} \ind_{i_s = n},\end{equation}
i.e., $m_t$ is the number of times the index $n$ was sampled in the first $t$ iterations of the phase. If we can show \eqref{eq:sensitivity} holds, then we are done because $m_t \le m$ by assumption.

We are left with proving \eqref{eq:sensitivity}, which we do by induction. The base case $t = 0$ is clear. Suppose \eqref{eq:sensitivity} holds at iteration $t$. In iteration $t + 1$, if $m_{t + 1} = m_t + 1$ (i.e., $i_t = n$ was sampled), then \eqref{eq:sensitivity} holds by the triangle inequality and the induction hypothesis, because all gradient steps $\eta \vg_t$ have $\norm{\eta \vg_t} \le \eta$, and projection to $\set$ can only decrease distances (Fact~\ref{fact:project_contract}). Otherwise, let $i_t = i \neq n$ be the sampled index, using the common point $\vx_i \in \calD \cap \calD'$. Now, \eqref{eq:sensitivity} follows from applying Lemma~\ref{lem:same_sgd_step} with
\[\vx \gets \vz_t - \vx_i,\; \vy \gets \vz'_t - \vx_i,\; \vu \gets \frac{\vz_t - \vx_i}{\norm{\vz_t - \vx_i}},\; \vv \gets \frac{\vz'_t - \vx_i}{\norm{\vz'_t - \vx_i}}.\]
In particular, we have that $\norm{\vx - \vy} = \Phi_t$, and $\norm{\vx' - \vy'} \ge \Phi_{t + 1}$ (due to Fact~\ref{fact:project_contract}), following notation from Lemma~\ref{lem:same_sgd_step}. We thus have $\Phi_{t + 1} \le \max(\Phi_t, 3\eta)$, which clearly also preserves \eqref{eq:sensitivity} inductively.
\end{proof}

We can now conclude our privacy proof by applying composition to Corollary~\ref{cor:phase_stable}.

\begin{lemma}\label{lem:sgd_private}
Let $\eps \in [0, 1]$. If $\frac 1 \rho \ge \frac{4\log(\frac 2 \delta)}{\eps^2} + \frac 2 \eps$, Algorithm~\ref{alg:phased_sgd} is $(\eps, \delta)$-DP.
\end{lemma}
\begin{proof}
We claim that Algorithm~\ref{alg:phased_sgd} satisfies $\rho$-CDP, conditioned on $\calI$ containing at most $m$ choices of $t$ (we denote this conditional variant by $\alg$). By applying the second part of Fact~\ref{fact:rdp} with $\alpha \gets \frac{2\log(\frac 2 \delta)}{\eps} + 1$, this implies that $\alg$ is $(\eps, \frac \delta 2)$-DP. Because $\alg$ has total variation distance at most $\frac \delta 8$ to Algorithm~\ref{alg:phased_sgd} due to Lemma~\ref{lem:m_small}, we conclude using Lemma~\ref{lem:dp_close} that Algorithm~\ref{alg:phased_sgd} is $(\eps, \delta)$-DP.

We are left to show $\alg$ satisfies $\rho$-CDP. In fact, we will show that for all $k \in [K]$, the output of the $k^{\text{th}}$ phase of $\alg$, i.e., $\hvx^{(k)}$, satisfies $(\frac {16} 9)^{-k} \cdot \frac \rho 2$-CDP (treating the starting iterate $\vz_0^{(k)} = \hvx^{(k - 1)}$ as fixed). Using composition of RDP (the first part of Fact~\ref{fact:rdp}), this implies $\alg$ is $\rho$-CDP as desired.

Finally, we bound the CDP of phase $k \in [K]$. Under $\alg$, we showed in Corollary~\ref{cor:phase_stable} that all iterates of the $k^{\text{th}}$ phase have sensitivity $\le (2m + 1)\eta^{(k)}$. Thus the average iterate $\bvx^{(k)}$ also has sensitivity $\le (2m + 1)\eta^{(k)}$ by the triangle inequality. We can now bound the CDP of the $k^{\text{th}}$ phase using the third part of Fact~\ref{fact:rdp}:
\[\frac{((2m + 1)\eta^{(k)})^2}{2(\sigma^{(k)})^2} =  \frac{((2m + 1)\eta)^2}{((2m + 1)\eta)^2} \cdot 16^{-k} \cdot 9^k  \cdot \frac \rho 2 \le \Par{\frac {16} {9}}^{-k} \cdot \frac \rho 2.\]
\end{proof}

\subsection{Utility of Algorithm~\ref{alg:phased_sgd}}\label{ssec:boost_util}

We now analyze the error guarantees for Algorithm~\ref{alg:phased_sgd} on optimizing the geometric median objective $f_{\calD}$ \eqref{eq:gm_def}. We begin by providing a high-probability bound on the utility guarantees of each single phase.

\begin{lemma}\label{lem:phase_util}
Following notation of Algorithm~\ref{alg:phased_sgd}, we have with probability $\ge 1 - \frac \delta 2$ that 
\[f_{\calD}(\bvx^{(1}) - f_{\calD}(\vxs) \le \frac{\hr^2}{2\eta^{(1)}T^{(1)}} + \frac{\eta^{(1)}}{2} + 4\hr\sqrt{\frac{2\log(\frac{4K}{\delta})}{T^{(1)}}},\]
where $\vxs \defeq \argmin_{\vx \in \ball^d(\bvx, \hr)} f(\vx)$, and
\[f_{\calD}(\bvx^{(k)}) - f_{\calD}(\bvx^{(k - 1)}) \le \frac{2(\sigma^{(k)})^2d\log(\frac{4K}{\delta})}{\eta^{(k)}T^{(k)}} + \frac{\eta^{(k)}}{2} + 8\sigma^{(k)}\log\Par{\frac{4K}{\delta}}\sqrt{\frac{2d}{T^{(k)}}} \text{ for all } 2 \le k \le K.\]
\end{lemma}
\begin{proof}
First, with probability $\ge 1 - \frac \delta {4}$, we have
\[\norm{\vxi^{(k)}} \le 2\sigma^{(k)}\sqrt{d\log\Par{\frac{4K}{\delta}}} \text{ for all } k \in [K],\]
by standard Gaussian concentration. Thus, $\bvx^{(k - 1)} \in \set^{(k)}$ for all $2 \le k \le K$, and $\vxs \in \set^{(1)}$, except with probability $\frac \delta 4$.
Next, consider the $k^{\text{th}}$ phase of Algorithm~\ref{alg:phased_sgd}, and for some $0 \le t < T_k$, let us denote 
\[\tvg_t^{(k)} \defeq \frac{\vz_t^{(k)} - \vx_i}{\norm{\vz_t^{(k)} - \vx_i} }\]
where $i \in [n]$ is the random index sampled on Line~\ref{line:sample_index} in the $t^{\text{th}}$ iteration of phase $k$. We also denote 
\[\vg_t^{(k)} \defeq \frac 1 n \sum_{i \in [n]} \frac{\vz_t^{(k)} - \vx_i}{\norm{\vz_t^{(k)} - \vx_i} } .\]
We observe that $\E[\tvg_t^{(k)}] = \vg_t^{(k)}$ for any realization of the randomness in all previous iterations. Now, by the standard Euclidean mirror descent analysis, see e.g., Theorem 3.2 of \cite{Bubeck15}, for any $\vu \in \set^{(k)}$,
\begin{equation}\label{eq:single_iter}
\inprod{\eta^{(k)} \tvg_t^{(k)}}{\vz_t^{(k)} - \vu} \le \frac{\norm{\vz_t^{(k)} - \vu}^2}{2} - \frac{\norm{\vz_{t + 1}^{(k)} - \vu}^2}{2} + \frac{(\eta^{(k)})^2}{2}.
\end{equation}
Here we implicitly used that $\norms{\tvg_t^{(k)}} \le 1$ for all choices of the sampled index $i \in [n]$. Now summing \eqref{eq:single_iter} for all iterations $0 \le t < T^{(k)}$, and normalizing by $\eta^{(k)}T^{(k)}$, we obtain
\begin{align*}
\frac 1 {T^{(k)}} \sum_{0 \le t < T^{(k)}} \inprod{\tvg_t^{(k)}}{\vz_t^{(k)} - \vu} \le \frac{\norm{\vz_0^{(k)} - \vu}^2}{2 \eta^{(k)} T^{(k)}} + \frac{\eta^{(k)}}{2}.
\end{align*}
Next, we claim that with probability $\ge 1 - \frac \delta 4$,
\begin{align*}\frac 1 {T^{(1)}} \sum_{0 \le t < T^{(1)}} \inprod{\vg_t^{(1)} - \tvg_t^{(1)}}{\vz_t^{(1)} - \vu} &\le 4\hr \sqrt{\frac{2\log\Par{\frac{4K}{\delta}}}{T^{(1)}}},\\
\frac 1 {T^{(k)}} \sum_{0 \le t < T^{(k)}} \inprod{\vg_t^{(k)} - \tvg_t^{(k)}}{\vz_t^{(k)} - \vu} &\le 8\sigma^{(k)}\log\Par{\frac{4K}{\delta}}\sqrt{\frac{2d}{T^{(k)}}}
 \text{ for all } 2 \le k \le K.\end{align*}
 In each case, this is because $\inprods{\vg_t^{(k)} - \tvg_t^{(k)}}{\vz_t^{(k)} - \vu}$ is a mean-zero random variable, that is bounded (with probability $1$) by twice the diameter of $\set^{(k)}$. Thus we can bound the sub-Gaussian parameter of their sum, and applying the Azuma-Hoeffding inequality then gives the result. Now, finally by convexity,
\begin{align*}
\frac 1 {T^{(1)}} \sum_{0 \le t < T^{(1)}} \inprod{\vg_t^{(1)}}{\vz_t^{(1)} - \vu} &\ge \frac 1 {T^{(1)}} \sum_{0 \le t < T^{(1)}} f_{\calD} (\vz^{(1)}_t) - f_{\calD}(\vu) \ge f_{\calD} (\bvx^{(1)}) - f_{\calD}(\vu), \\
\frac 1 {T^{(k)}} \sum_{0 \le t < T^{(k)}} \inprod{\vg_t^{(k)}}{\vz_t^{(k)} - \vu} &\ge \frac 1 {T^{(k)}} \sum_{0 \le t < T^{(k)}} f_{\calD} (\vz^{(k)}_t) - f_{\calD}(\vu) \ge f_{\calD} (\bvx^{(k)}) - f_{\calD}(\vu) \text{ for all } 2 \le k \le K.
\end{align*}
 Combining the above three displays, and plugging in $\vu \gets \vxs$ or $\vu \gets \bvx^{(k - 1)}$, now gives the conclusion.
\end{proof}

By summing the conclusion of Lemma~\ref{lem:phase_util} across all phases, we obtain an overall error bound.

\begin{lemma}\label{lem:util_boost}
Following notation of Algorithm~\ref{alg:phased_sgd} and Lemma~\ref{lem:phase_util}, we have with probability $\ge 1 - \delta$ that
\[f_{\calD}(\hvx^{(K)}) - f_{\calD}(\vxs) \le \frac{\hr^2}{16 \eta T} + 19\eta +8\hr\sqrt{\frac{\log(\frac{4K}{\delta})}{T}}  + \frac{1314 T \eta d \log^4(\frac{8K}{\delta})}{\rho n^2}.\]
\end{lemma}
\begin{proof}
Throughout this proof, condition on the conclusion of Lemma~\ref{lem:phase_util} holding, as well as
\[\norm{\vxi^{(K)}} \le 2\sigma^{(K)} \sqrt{d\log\Par{\frac 2 \delta}},\]
both of which hold with probability $\ge 1 - \delta$ by a union bound. Next, by Lemma~\ref{lem:phase_util},
\begin{align*}
f_{\calD}(\hvx_K) - f_{\calD}(\vxs)&= f_{\calD}(\bvx^{(1)}) - f_{\calD}(\vxs) + \sum_{k = 2}^K f_{\calD}(\bvx^{(k)}) - f_{\calD}(\bvx^{(k - 1)}) + f_{\calD}(\hvx_K) - f_{\calD}(\bvx_K) \\
&\le \frac{\hr^2}{2\eta^{(1)}T^{(1)}} + \frac{\eta^{(1)}}{2} + 4\hr\sqrt{\frac{2\log(\frac{4K}{\delta})}{T^{(1)}}} \\
&+ \sum_{k = 2}^K \Par{\frac{2(\sigma^{(k)})^2d\log(\frac{4K}{\delta})}{\eta^{(k)}T^{(k)}} + \frac{\eta^{(k)}}{2} + 8\sigma^{(k)}\log\Par{\frac{4K}{\delta}}\sqrt{\frac{2d}{T^{(k)}}}} + \norm{\vxi^{(K)}} \\
&\le \frac{\hr^2}{16 \eta T} + \frac \eta 2 + 8\hr\sqrt{\frac{\log(\frac{4K}{\delta})}{T}} + \frac{144m^2 \eta d \log(\frac{4K}{\delta})}{\rho T} + \frac{12\sqrt{d}m\eta\log(\frac{4K}{\delta})}{\sqrt{\rho T}}  \\
&\le \frac{\hr^2}{16 \eta T} + \frac \eta 2 + 8\hr\sqrt{\frac{\log(\frac{4K}{\delta})}{T}} + \frac{1296 T \eta d \log^3(\frac{8K}{\delta})}{\rho n^2} + \frac{36\sqrt{dT}\eta\log^2(\frac{8K}{\delta})}{\sqrt{\rho} n}  \\
&\le \frac{\hr^2}{16 \eta T} + 19\eta +8\hr\sqrt{\frac{\log(\frac{4K}{\delta})}{T}}  + \frac{1314 T \eta d \log^4(\frac{8K}{\delta})}{\rho n^2}.
\end{align*}
The third line used Lipschitzness of $f_{\calD}$, the fourth summed parameters using various geometric sequences, the fifth plugged in our value of $m$, and the last split the fifth term using $2ab \le a^2 + b^2$ appropriately.
\end{proof}

By combining Lemmas~\ref{lem:sgd_private} and~\ref{lem:util_boost} with Theorem~\ref{thm:constant_factor}, we obtain our main result on privately approximating the geometric median to an arbitrary multiplicative factor $1 + \alpha$, given enough samples.

\begin{theorem}\label{thm:boost}
Let $\calD = \{\vx_i\}_{i \in [n]} \subset \ball^d(R)$ for $R > 0$, let $0 < r \le 4r^{(0.9)}(\calD)$, and let $(\alpha, \eps, \delta) \in [0, 1]^3$. Suppose that
\[n \ge C \cdot\Par{\frac{\sqrt d}{\alpha\eps}\log^{2.5}\Par{\frac{\log(\frac d {\alpha\delta\eps})}{\delta}}},\]
for a sufficiently large constant $C$. There is an $(\eps, \delta)$-DP algorithm (Algorithm~\ref{alg:phased_sgd} using Theorem~\ref{thm:constant_factor} to compute the parameters $(\bvx, \hr)$) that returns $\hvx$ such that with probability $\ge 1 - \delta$, following notation \eqref{eq:gm_def},
\[f_{\calD}(\hvx) \le (1 + \alpha)f_{\calD}(\vxs(\calD)).\]
The algorithm runs in time
\[O\Par{nd\log\Par{\frac R r}\log\Par{\frac{d\log(\frac R r)}{\alpha\delta\eps}} + \frac{d}{\alpha^2}\log\Par{\frac {\log(\frac d {\alpha\delta\eps})} {\delta}}}.\]
\end{theorem}
\begin{proof}
We first apply Theorem~\ref{thm:constant_factor} to compute a $(\bvx, \hr)$ pair satisfying
$
\hr \le C' r^{(0.9)}(\calD)$, $\norm{\bvx - \vxs(\calD)} \le \hr$,
for a universal constant $C'$,
subject to $(\frac \eps 2, \frac \delta 2)$-DP and $\frac \delta 2$ failure probability. We can verify that Theorem~\ref{thm:constant_factor} gives these guarantees within the stated runtime, for a large enough $C$. Next, we call Algorithm~\ref{alg:phased_sgd} with $\rho \gets \frac{\eps^2}{32\log(\frac 4 \delta)}$ and $\delta \gets \frac \delta 2$, which is $(\frac \eps 2, \frac \delta 2)$-DP by Lemma~\ref{lem:sgd_private}, so this composition is $(\eps, \delta)$-DP. 

Denoting $\hvx \defeq \hvx^{(K)}$ to be the output of Algorithm~\ref{alg:phased_sgd}, Lemma~\ref{lem:util_boost} guarantees that with probability $\ge 1 - \frac \delta 2$, 
\[f_{\calD}(\hvx) - f_{\calD}(\vxs(\calD)) \le \frac{\hr^2}{16\eta T} + 19\eta + 8\hr\sqrt{\frac{\log(\frac{16K}{\delta})}{T}} + \frac{5256T\eta d\log^5(\frac{16K}{\delta})}{\eps^2 n^2}, \]
for some choice of $\eta, T$ and our earlier choices of privacy parameters. Optimizing in $\eta$, we have
\begin{align*}
f_{\calD}(\hvx) - f_{\calD}(\vxs(\calD)) \le 12\hr\sqrt{\frac{\log(\frac{16K}{\delta})}{T}} + \frac{37 \hr \sqrt{d\log^5(\frac{16K}{\delta})}}{\eps n}.
\end{align*}
Finally, for a large enough $C$ in the definition of $n$, and $T \ge n + \frac{57600 (C')^2\log(\frac{16K}{\delta})}{\alpha^2}$, we obtain
\[f_{\calD}(\hvx) - f_{\calD}(\vxs(\calD)) \le 12\hr\sqrt{\frac{\log(\frac{16K}{\delta})}{T}} + \frac{37 \hr \sqrt{d\log^5(\frac{16K}{\delta})}}{\eps n} \le \frac{\alpha \hr}{10C'} \le \frac{\alpha r^{(0.9)}}{10} \le \alpha f_{\calD}(\vxs(\calD)).\]
The last inequality used Lemma~\ref{lem:opt_lb}. Now, the runtime follows from combining Theorem~\ref{thm:constant_factor} and the fact that every iteration of Algorithm~\ref{alg:phased_sgd} can clearly be implemented in $O(d)$ time.
\end{proof}

%% file: experiments.tex
\section{Experiments}\label{sec:experiments}

In this section, we present empirical evidence supporting the efficacy of our techniques. We implement and conduct experiments on Algorithm~\ref{alg:frf} (the radius estimation step of Section~\ref{sec:constant}) and Algorithm~\ref{alg:phased_sgd}, to evaluate how subsampled estimates and DP-SGD respectively improve the performance of our algorithm.\footnote{Our subsampling experiments were performed on a single Google Colab CPU, and our boosting experiments were performed on a personal Apple M4 with 16GB RAM.}

We do not present experiments on Algorithm~\ref{alg:fc}, as our analysis results in loose constants, which in our preliminary experimentation significantly impacted its performance in practice. We leave optimizing the performance of this step as an important step for future work. In our experiments, Algorithm~\ref{alg:phased_sgd} was fairly robust to the choice of initialization, so it is possible that private heuristics may serve as stand-in to this step. Moreover, our Algorithm~\ref{alg:fc} and Section 2.2 of \cite{haghifam2024private} had essentially the same runtime, so we find it in line with our conceptual contribution to focus on evaluating the other two components.

We use two types of synthetic datasets with outliers, described here. To avoid contamination in hyperparameter selection, every experiment is performed with a freshly-generated dataset.

$\gaussiancluster(R, n, d, \sigma, \mathrm{frac}_{\mathrm{in}})$: This dataset is described in Appendix H of \cite{haghifam2024private}. We draw $n_{\mathrm{in}}=\mathrm{frac}_{\mathrm{in}}\,n$ points i.i.d.\ from $\mathcal{N}(\vmu,\sigma^2 \id_d)$ with $\vmu$ uniform on the sphere of radius $\frac R 2$, and $n_{\mathrm{out}}=n-n_{\mathrm{in}}$ outliers uniformly from the Euclidean ball of radius $R$. 

$\heavytailed(\nu, n, d)$: This dataset samples $n$ points in $\mathbb{R}^{d}$ from a zero‐mean multivariate Student’s $t$ distribution with identity scale and degrees of freedom $\nu$.

\subsection{Subsampling}\label{ssec:subsampling}

In this section,\footnote{We provide code for the experiments in this section \href{https://colab.research.google.com/drive/1kPaOIuGEcYJtlBERSxI4pkVr6FYHcM2z?usp=sharing}{here}.} we describe our experiments to show the benefit of subsampling in $\mathsf{FastRadius}$ (Algorithm~\ref{alg:frf}) over $\mathsf{RadiusFinder}$ (Algorithm 1 from \cite{haghifam2024private}) for differentially private estimation of the quantile radius, which is the first step in differentially private estimation of the geometric median. 

In the first experiment (Figure~\ref{fig:exp_gaussian_cluster}), we
set $n=1000$, $d=10$, inlier fraction $\mathrm{frac}_{\mathrm{in}}=0.9$, standard deviation $\sigma=0.1$, and choose an upper bound $R$ from the set $\{0.5,1,2,4,8,10\}$. The dataset is generated as $\gaussiancluster(R, n, d, \sigma,\mathrm{frac}_{\mathrm{in}})$  dataset. We set privacy parameters $\varepsilon=1.0$ and $\delta=10^{-5}$, quantile fraction $\gamma=0.75$,  and for each trial we sample $r_{\min}\sim\Unif([0.005,0.02])$ to randomly initialize our search grid. Since $\gamma < \mathrm{frac}_{\mathrm{in}}$, we estimate the ground‐truth quantile radius $r_{\mathrm{true}} = \sigma\sqrt{d}$, run both algorithms on this dataset, measure the estimated radius $\hat r$ and wall‐clock runtime, and report the mean and standard deviation of the estimation ratio $\hat r/r_{\mathrm{true}}$ and runtime over 100 independent trials.

In the second experiment (Figure~\ref{fig:exp_student_t_dist}), we assess the robustness of $\mathsf{FastRadius}$ and $\mathsf{RadiusFinder}$ to heavy‐tailed data. We set $n=1000$ and $d = 10$ in the $\heavytailed(\nu, n, d)$ dataset with varying degrees of freedom $\nu\in\{2,  4,  6,  8, 10, 12, 14, 16, 18, 20\}$. For each trial,  we sample $r_{\min}\sim\Unif([0.005,0.02])$, set privacy parameters $\varepsilon=1.0$, $\delta=10^{-5}$ and quantile fraction $\gamma=0.75$. We estimate the theoretical quantile radius $r_{\mathrm{true}}=\sqrt{d\,F_{d,\nu}(\gamma)}$ where $F_{d,\nu}$ is the CDF of an $\mathrm{F}(d,\nu)$ distribution, which is the Fisher F-distribution with $d$ and $\nu$ degrees of freedom, execute both algorithms, record $\hat r$ and runtime, and summarize the mean and standard deviation of the ratio $\hat r/r_{\mathrm{true}}$ and runtime across 100 repetitions.

\begin{figure}[htb]
  \centering

  \begin{subfigure}[t]{0.45\textwidth}
    \centering
    \includegraphics[width=\textwidth]{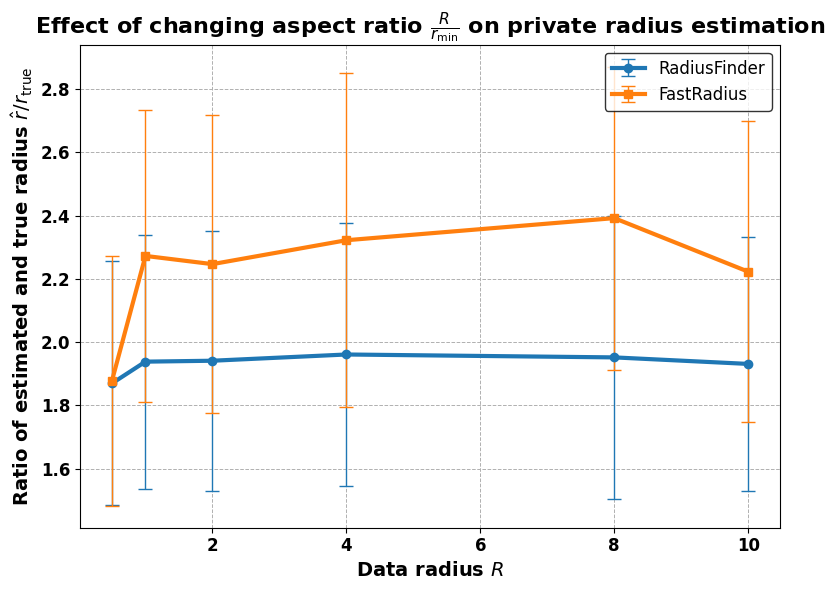}
    \caption{Ratio of the estimated quantile radius to the true radius with varying data radius \(R\).}
    \label{fig:exp_gaussian_cluster}
  \end{subfigure}
  \hfill
  \begin{subfigure}[t]{0.45\textwidth}
    \centering
    \includegraphics[width=\textwidth]{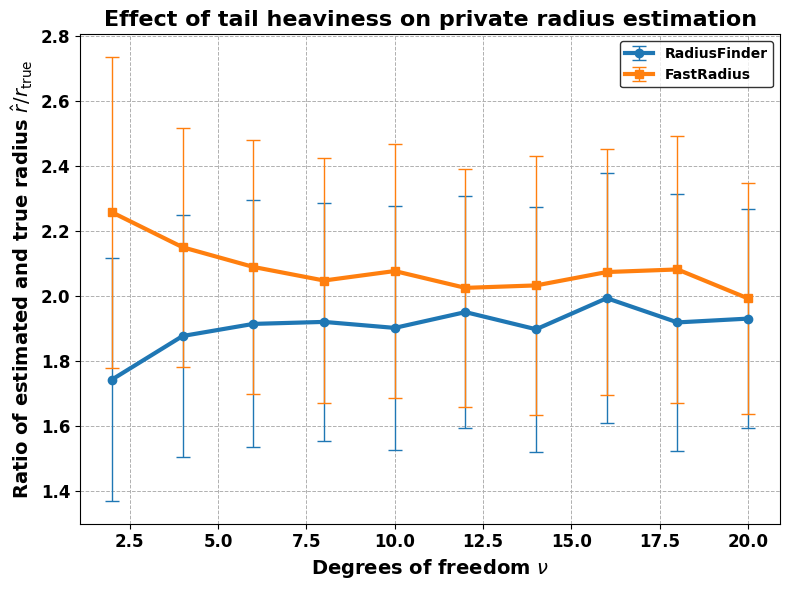}
    \caption{\label{fig:exp_student_t_dist} Ratio of the estimated quantile radius to the true radius with varying degrees of freedom \(\nu\).}
  \end{subfigure}
  \caption{\label{fig:exp_subsampling} Comparison of $\mathsf{RadiusFinder}$ and $\mathsf{FastRadius}$ across different data distributions. Plots averaged across $100$ trials and standard deviations are reported as error bars.}
\end{figure}

We observe that in both cases, across a range of increasingly heavier tails of the distributions, both algorithms achieve reasonable approximation to the true quantile radius, always staying multiplicatively between roughly 1.2 to 3 of the true quantile radius. We further record the average wall-clock time required by both algorithms in Table~\ref{tab:wall_clock_times_subsampling}.  We observe that $\mathsf{FastRadius}$ is significantly faster compared to $\mathsf{RadiusFinder}$, while performing competitively in terms of estimation quality. 

We remark that we also experimented with varying $n \in \{500, 1000, 2000\}$ and $d \in \left\{5, 10, 20\right\}$ and observed qualitatively similar trends for the performance of both algorithms. 

\begin{table}[h]
  \centering
  \caption{\label{tab:wall_clock_times_subsampling}Average wall‐clock time in seconds over 100 trials for each algorithm in each experiment}
  \label{tab:avg-runtime}
  \begin{tabular}{lcc}
    \toprule
    Experiment & $\mathsf{RadiusFinder}$ & $\mathsf{FastRadius}$ \\
    \midrule
    Varying $R$ (Figure~\ref{fig:exp_gaussian_cluster})       & $1.192 \pm 0.047$   & $0.0411 \pm 0.002$   \\
    Varying $\nu$ (Figure~\ref{fig:exp_student_t_dist})     & $1.204 \pm 0.171$ & $0.0413 \pm 0.007$ \\
    \bottomrule
  \end{tabular}
\end{table}

\subsection{Boosting}

In this section,\footnote{We provide code for the experiments in this section \href{https://colab.research.google.com/drive/1bEDNwE5Bge-FABqcbXDLwYcpfzv7tJGX?usp=sharing}{here}.} we evaluate the performance of our boosting algorithm in Section~\ref{sec:boosting} based on a low-pass DP-SGD implementation, compared to the baseline method from \cite{haghifam2024private}. We will in fact evaluate three methods: (1) the baseline method, $\DPGD$ (vanilla DP gradient descent), as described in Algorithms 3 and 6, \cite{haghifam2024private}, but with an optimized step size selected through ablation studies; (2) $\StableDPSGD$, i.e., our Algorithm~\ref{alg:phased_sgd} implemented as written, and (3) $\FixedOrderDPSGD$, a variant of our Algorithm~\ref{alg:phased_sgd} with the last optimization described in Remark~\ref{rem:practical}. We calibrated our noise level in $\FixedOrderDPSGD$ to ensure a fixed level of CDP via a group privacy argument, where we use that each dataset element is deterministically accessed at most $m = \lceil \frac T n \rceil$ times in $\FixedOrderDPSGD$ with $T$ iterations. 

We next describe our hyperparameter optimization for the baseline, $\DPGD$, as implemented in Algorithm 6, \cite{haghifam2024private}. We note that this implementation of $\DPGD$ satisfies $\rho$-CDP for an arbitrary choice of step size $\eta$ (as it scales the noise appropriately), so we are free to tune for the best choice of $\eta$.

Algorithm 3 in \cite{haghifam2024private} recommends a constant step size of $\eta_\mathrm{base}=2\hat{r}\sqrt{\frac{d}{6\rho n^2}}$, where $\hat{r}$ is the estimated radius. However, conventional analyses of projected gradient descent (cf.\ Section 3.1, \cite{Bubeck15}) recommend a step size scaling as a multiple of $\eta_\mathrm{base}=\hat{r} \cdot \frac 1 {\sqrt{T}}$, where $T$ is the iteration count. Moreover, there is theoretical precedent for DP-(S)GD going through a phase transition in step sizes for different regimes of $n$, $T$ (e.g., \cite{FeldmanKT20}, Theorem 4.4). We thus examined multiples of both of these choices of $\eta_{\mathrm{base}}$, i.e., we used step sizes $\eta=\eta_\mathrm{base}\cdot\eta_\mathrm{multiplier}$ with multipliers $\eta_\mathrm{multiplier} \in \{0.25, 0.5, 0.75, 1, 1.25, 1.5, 1.75\}$ and $\eta_\mathrm{base} \in \{2\hat{r}\sqrt{\frac{d}{6\rho n^2}}, \hat{r} \cdot \frac 1 {\sqrt{T}}\}$ in an ablation study, across all datasets appearing in our experiments. 

Our results indicated that the $\eta_\mathrm{multiplier}$ depends significantly on dataset size, in that larger datasets benefit from higher $\eta_\mathrm{multiplier}$. We observed that if we chose $\eta_\mathrm{base}=2\hat{r}\sqrt{\frac{d}{6\rho n^2}}$, as $n$ increases, using larger $\eta_\mathrm{multiplier}$ values consistently reduced optimization error, but with diminishing returns. However, using $\hat{r} \cdot \frac 1 {\sqrt{T}}$ as the base step size yielded significantly more stable performance across multiple scales of $n, T$, compared to the recommendation in \cite{haghifam2024private}, suggesting this is the correct scaling in practice. Our findings were that $\DPGD$ yielded the consistently best performance with $\eta_{\textup{multiplier}} = 1$ and $\eta_{\textup{base}} = \hat{r} \cdot \frac 1 {\sqrt{T}}$.\footnote{Larger step size multipliers yielded better performance on $\gaussiancluster$ data, but led to large amounts of instability on $\heavytailed$ data. We chose the largest multiplier that did not result in significant instability on any dataset.}


We now describe our setup. In all our experiments, we set $d = 50$, $\rho = 0.5$, and vary $n \in \{100, 1000, 10000\}$. For the $\gaussiancluster$ dataset, we set $\sigma = 0.1$ and vary the bounding radius $R \in \{25, 50, 100\}$. We set our estimated initial radius $\hr = 20\sigma \sqrt{d}$ and initialize all algorithms at a uniformly random point on the surface of $\ball^d(0.75\hr)$.\footnote{We chose a relatively pessimistic multiple of $\hr$ to create a larger initial loss and account for estimation error.} For the $\heavytailed$ dataset, we use the same values of $d$, $\rho$, and the same range of $n$. We vary $\nu \in \{2.5, 5.0, 10.0\}$, set our estimated initial radius $\hr = 20\sqrt{dF_{d,\nu}(0.75)}$ to be consistent with Section~\ref{ssec:subsampling}, and again initialize randomly on the surface of $\ball^d(0.75\hr)$.

In our first set of experiments (Figures~\ref{fig:gc_sgd} and~\ref{fig:ht_sgd}), we used the middle ``scale'' parameter, i.e., $R = 50$ for the $\gaussiancluster$ dataset and $\nu = 5.0$ for the $\heavytailed$ dataset, varying $n$ only. We report the performance of the three evaluated methods, plotting the passes over the dataset used by the excess error. Our error metric is $\frac 1 {\hr} \cdot (f_{\calD}(\hvx) - f_{\calD}(\hvx))$, i.e., a multiple of the ``effective radius'' used in the experiment. This is a more reflective performance metric than the corresponding multiple of $f_{\calD}$, as our algorithms achieve this bound (see discussion after Theorem~\ref{thm:constant_factor}), and $f_{\calD} \ge \hr$ for our datasets due to outliers. 

\begin{figure}[htb]
  \centering

  \begin{subfigure}[t]{0.3\textwidth}
    \centering
    \includegraphics[width=\textwidth]{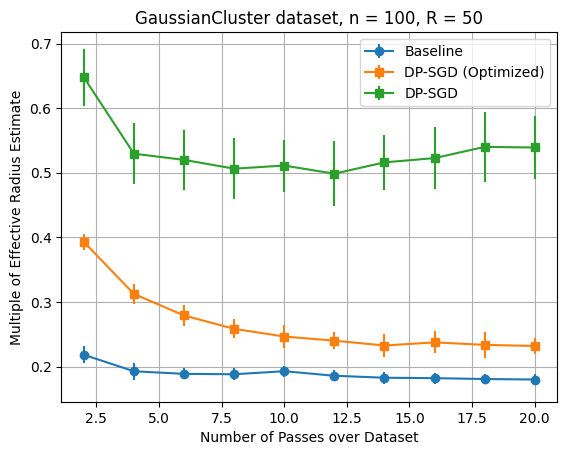}
    \caption{$n = 100$}
  \end{subfigure}
  \hfill
  \begin{subfigure}[t]{0.3\textwidth}
    \centering
    \includegraphics[width=\textwidth]{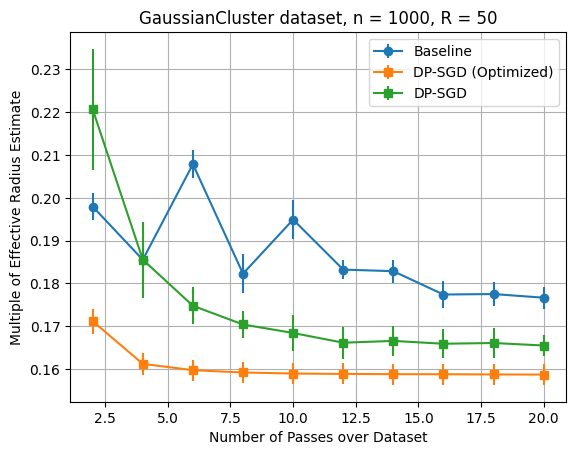}
    \caption{$n = 1000$}
  \end{subfigure}
  \hfill
  \begin{subfigure}[t]{0.3\textwidth}
    \centering
    \includegraphics[width=\textwidth]{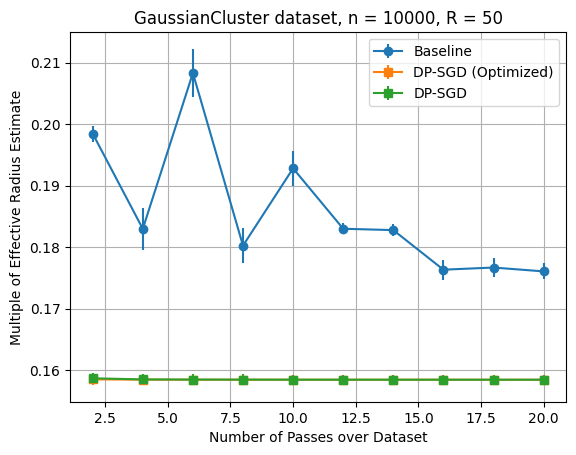}
    \caption{$n = 10000$}
  \end{subfigure}
  \caption{\label{fig:gc_sgd} Comparison of $\DPGD$, $\StableDPSGD$, and $\FixedOrderDPSGD$ across $\gaussiancluster$ data over $\R^{50}$, varying $n$. Plots averaged across $20$ trials and standard deviations are reported as error bars.}
\end{figure}

Across $\gaussiancluster$ datasets of size $n \in \{100, 1000, 10000\}$, we found that $\FixedOrderDPSGD$ consistently outperformed $\StableDPSGD$, and consistently outperformed the baseline by a significant margin once dataset sizes were large enough. As the theory predicts, the gains of stochastic methods in terms of error-to-pass ratios are more stark when dataset sizes are larger, reflecting the superlinear gradient query complexity (each requiring one pass) that $\DPGD$ needs to obtain the optimal utility.

\begin{figure}[htb!]
  \centering
  \begin{subfigure}[t]{0.3\textwidth}
    \centering
    \includegraphics[width=\textwidth]{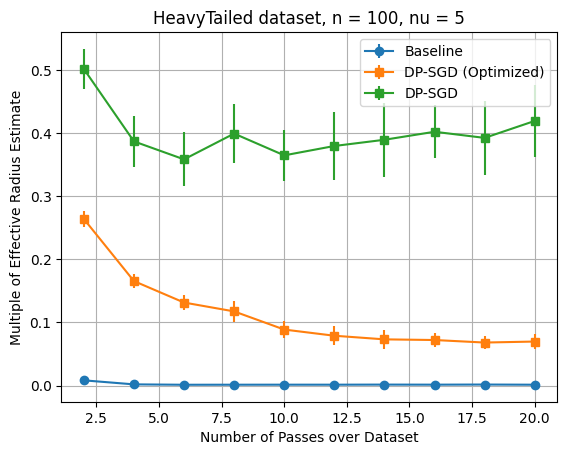}
    \caption{$n = 100$}
  \end{subfigure}
  \hfill
  \begin{subfigure}[t]{0.3\textwidth}
    \centering
    \includegraphics[width=\textwidth]{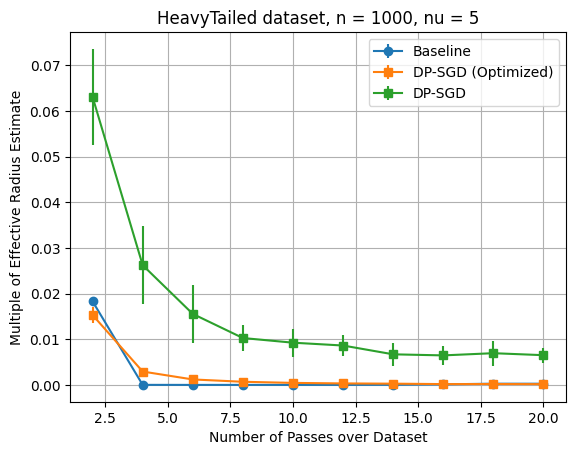}
    \caption{$n = 1000$}
  \end{subfigure}
  \hfill
  \begin{subfigure}[t]{0.3\textwidth}
    \centering
    \includegraphics[width=\textwidth]{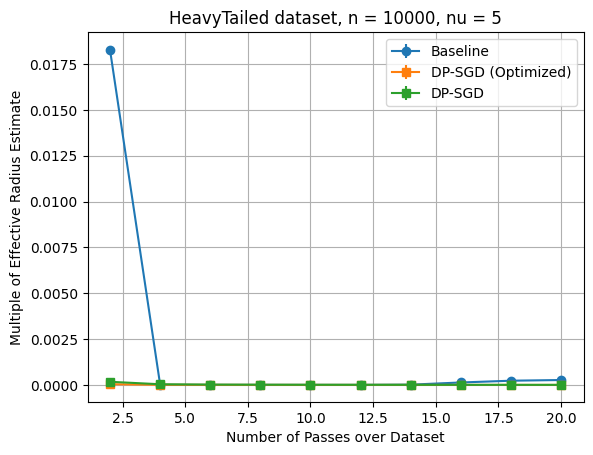}
    \caption{$n = 10000$}
  \end{subfigure}
  \caption{\label{fig:ht_sgd} Comparison of $\DPGD$, $\StableDPSGD$, and $\FixedOrderDPSGD$ across $\heavytailed$ data over $\R^{50}$, varying $n$. Plots averaged across $20$ trials and standard deviations are reported as error bars.}
\end{figure}

We next present our comparisons for the $\heavytailed$ dataset. Again, our (optimized) method led to similar or better performance than the baseline for larger $n$. We suspect that the improved performance of the baseline owes to the relative ``simplicity'' of this dataset, e.g., it is rotationally symmetric around the population geometric median, and this is likely to be reflected in a sample.

In our second set of experiments (Figures~\ref{fig:gc_scale_sgd} and~\ref{fig:ht_scale_sgd}), we fixed the size of the dataset at $n = 1000$, varying the scale parameter ($R$ for $\gaussiancluster$ and $\nu$ for $\heavytailed$). The relative performance of our evaluated algorithms was essentially unchanged across the parameter settings we considered.

\begin{figure}[htb]
  \centering

  \begin{subfigure}[t]{0.3\textwidth}
    \centering
    \includegraphics[width=\textwidth]{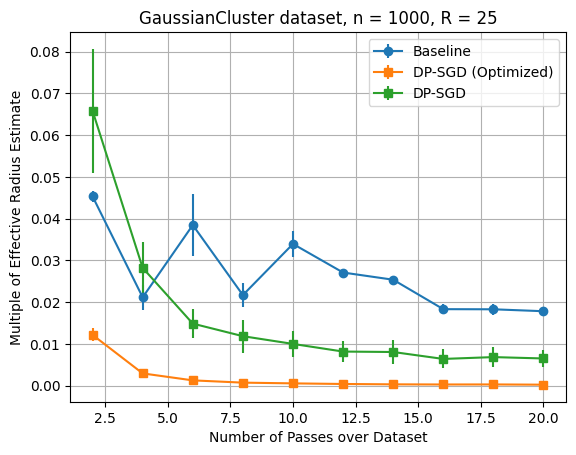}
    \caption{$R = 25$}
  \end{subfigure}
  \hfill
  \begin{subfigure}[t]{0.3\textwidth}
    \centering
    \includegraphics[width=\textwidth]{clusterdpsgd/GaussianCluster_dataset,_n=1000,_R=50.png}
    \caption{$R = 50$}
  \end{subfigure}
  \hfill
  \begin{subfigure}[t]{0.3\textwidth}
    \centering
    \includegraphics[width=\textwidth]{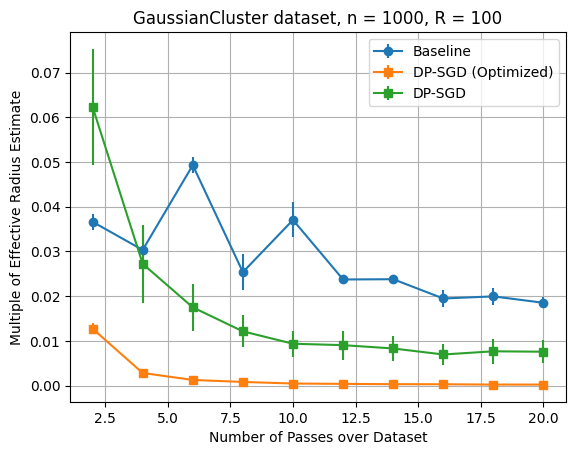}
    \caption{$R = 100$}
  \end{subfigure}
  \caption{\label{fig:gc_scale_sgd} Comparison of $\DPGD$, $\StableDPSGD$, and $\FixedOrderDPSGD$ across $\gaussiancluster$ data over $\R^{50}$, varying $R$. Plots averaged across $20$ trials and standard deviations are reported as error bars.}
\end{figure}

\begin{figure}[htb]
  \centering

  \begin{subfigure}[t]{0.3\textwidth}
    \centering
    \includegraphics[width=\textwidth]{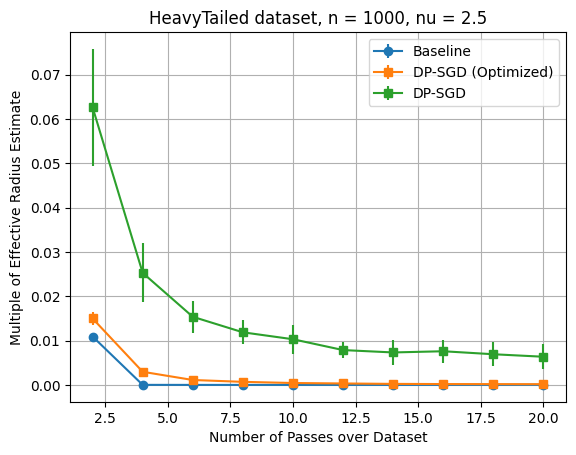}
    \caption{$\nu = 2.5$}
  \end{subfigure}
  \hfill
  \begin{subfigure}[t]{0.3\textwidth}
    \centering
    \includegraphics[width=\textwidth]{HeavyTailedsgd/HeavyTailed_dataset,_n=1000,_nu=5.png}
    \caption{$\nu = 5$}
  \end{subfigure}
  \hfill
  \begin{subfigure}[t]{0.3\textwidth}
    \centering
    \includegraphics[width=\textwidth]{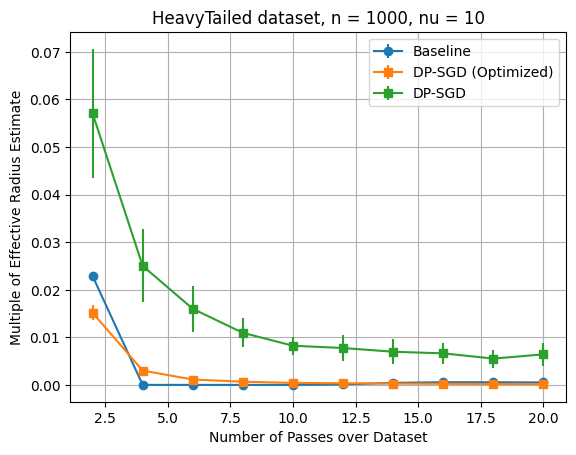}
    \caption{$\nu = 10$}
  \end{subfigure}
  \caption{\label{fig:ht_scale_sgd} Comparison of $\DPGD$, $\StableDPSGD$, and $\FixedOrderDPSGD$ across $\heavytailed$ data over $\R^{50}$, varying $\nu$. Plots averaged across $20$ trials and standard deviations are reported as error bars.}
\end{figure}

Finally, we remark that one major limitation of our evaluation is that full-batch gradient methods such as $\DPGD$ can be implemented with parallelized gradient computations, leading to wall-clock time savings. In our experiments, $\DPGD$ often performed better than $\FixedOrderDPSGD$ in terms of wall-clock time (for the same estimation error), even when it incurred significantly larger pass complexities. On the other hand, we expect the gains of methods based on DP-SGD to be larger as the dataset size and dimension $(n, d)$ grow. There are interesting natural extensions towards realizing the full potential of private optimization algorithms in practice, such as our Algorithm~\ref{alg:phased_sgd},  e.g., the benefits of using adaptive step sizes or minibatches, which we believe are important and exciting future directions.

%% file: appendix.tex
\appendix

\input{hsu_discuss}

%% file: hsu_discuss.tex
\section{Discussion of \cite{haghifam2024private} runtime}\label{app:hsu_discuss}

We give a brief discussion of the runtime of the \cite{haghifam2024private} algorithm in this section, as the claimed runtimes in the original paper do not match those described in Section~\ref{sec:introduction}. As the \cite{haghifam2024private} algorithm is split into three parts (the first two of which correspond to the warm start phase and the third of which corresponds to the boosting phase), we discuss the runtime of each part separately.

\paragraph{Radius estimation.} The radius estimation component of \cite{haghifam2024private} corresponds to Algorithm 1 and Section 2.1 of the paper. The authors claim a runtime of $O(n^2 \log(\frac R r))$ for this step due to need to do pairwise distance comparisons on a dataset of size $n$, for $O(\log(\frac R r))$ times in total. However, we believe the runtime of this step should be $O(n^2 d \log(\frac R r))$, accounting for the $O(d)$ cost of each comparison.

\paragraph{Centerpoint estimation.} The centerpoint estimation component of \cite{haghifam2024private} corresponds to Algorithm 2 and Section 2.2 of the paper. We agree with the authors that this algorithm runs in time $O(nd\log(\frac R r))$, and in particular, this step does not dominate any runtime asymptotically.

\paragraph{Boosting.} The boosting component of \cite{haghifam2024private} corresponds to Algorithms 3 and 4 and Section 3 of the paper. The authors provide two different boosting procedures (based on gradient descent and cutting-plane methods) and state their runtimes as $\tO(n^2 d)$ and $\tO(nd^2 + d^{2 + \omega})$, where $\omega < 2.372$ is the current matrix multiplication exponent \cite{AlmanDWXXZ25}. We agree with the runtime analysis of the cutting-plane method; however, we believe there is an additive $\tO(n^3\eps^2)$ term in the runtime of gradient descent. This follows by noting that Algorithm 3 uses $\approx \frac{n^2\eps^2}{d}$ iterations, each of which takes $O(nd)$ time to implement.

%% file: main.bbl
\newcommand{\etalchar}[1]{$^{#1}$}
\begin{thebibliography}{BFTGT19}

\bibitem[ACG{\etalchar{+}}16]{abadi2016deep}
Martin Abadi, Andy Chu, Ian Goodfellow, H~Brendan McMahan, Ilya Mironov, Kunal Talwar, and Li~Zhang.
\newblock Deep learning with differential privacy.
\newblock In {\em Proceedings of the 2016 ACM SIGSAC conference on computer and communications security}, pages 308--318, 2016.

\bibitem[AD20]{asi2020instance}
Hilal Asi and John~C Duchi.
\newblock Instance-optimality in differential privacy via approximate inverse sensitivity mechanisms.
\newblock {\em Advances in neural information processing systems}, 33:14106--14117, 2020.

\bibitem[ADV{\etalchar{+}}25]{AlmanDWXXZ25}
Josh Alman, Ran Duan, Virginia {Vassilevska Williams}, Yinzhan Xu, Zixuan Xu, and Renfei Zhou.
\newblock More asymmetry yields faster matrix multiplication.
\newblock In {\em Proceedings of the 2025 Annual {ACM-SIAM} Symposium on Discrete Algorithms, {SODA} 2025}, pages 2005--2039. {SIAM}, 2025.

\bibitem[AFKT21]{asi2021private}
Hilal Asi, Vitaly Feldman, Tomer Koren, and Kunal Talwar.
\newblock Private stochastic convex optimization: Optimal rates in l1 geometry.
\newblock In {\em International Conference on Machine Learning}, pages 393--403. PMLR, 2021.

\bibitem[AL22]{ashtiani2022private}
Hassan Ashtiani and Christopher Liaw.
\newblock Private and polynomial time algorithms for learning gaussians and beyond.
\newblock In {\em Conference on Learning Theory}, pages 1075--1076. PMLR, 2022.

\bibitem[AL23]{asi2023user}
Hilal Asi and Daogao Liu.
\newblock User-level differentially private stochastic convex optimization: Efficient algorithms with optimal rates.
\newblock {\em arXiv preprint arXiv:2311.03797}, 2023.

\bibitem[ALT24]{AsiLT24}
Hilal Asi, Daogao Liu, and Kevin Tian.
\newblock Private stochastic convex optimization with heavy tails: Near-optimality from simple reductions.
\newblock In {\em Advances in Neural Information Processing Systems 38: Annual Conference on Neural Information Processing Systems 2024}, 2024.

\bibitem[ATMR21]{andrew2021differentially}
Galen Andrew, Om~Thakkar, Brendan McMahan, and Swaroop Ramaswamy.
\newblock Differentially private learning with adaptive clipping.
\newblock {\em Advances in Neural Information Processing Systems}, 34:17455--17466, 2021.

\bibitem[BD14]{barber2014privacy}
Rina~Foygel Barber and John~C Duchi.
\newblock Privacy and statistical risk: Formalisms and minimax bounds.
\newblock {\em arXiv preprint arXiv:1412.4451}, 2014.

\bibitem[BDKU20]{BiswasDKU20}
Sourav Biswas, Yihe Dong, Gautam Kamath, and Jonathan~R. Ullman.
\newblock Coinpress: Practical private mean and covariance estimation.
\newblock In {\em Advances in Neural Information Processing Systems 33: Annual Conference on Neural Information Processing Systems 2020}, 2020.

\bibitem[BFGT20]{bassily2020stability}
Raef Bassily, Vitaly Feldman, Crist{\'o}bal Guzm{\'a}n, and Kunal Talwar.
\newblock Stability of stochastic gradient descent on nonsmooth convex losses.
\newblock {\em Advances in Neural Information Processing Systems}, 33:4381--4391, 2020.

\bibitem[BFTGT19]{bassily2019private}
Raef Bassily, Vitaly Feldman, Kunal Talwar, and Abhradeep Guha~Thakurta.
\newblock Private stochastic convex optimization with optimal rates.
\newblock {\em Advances in neural information processing systems}, 32, 2019.

\bibitem[BGN21]{bassily2021non}
Raef Bassily, Crist{\'o}bal Guzm{\'a}n, and Anupama Nandi.
\newblock Non-euclidean differentially private stochastic convex optimization.
\newblock In {\em Conference on Learning Theory}, pages 474--499. PMLR, 2021.

\bibitem[BGS{\etalchar{+}}21]{brown2021covariance}
Gavin Brown, Marco Gaboardi, Adam Smith, Jonathan Ullman, and Lydia Zakynthinou.
\newblock Covariance-aware private mean estimation without private covariance estimation.
\newblock {\em Advances in neural information processing systems}, 34:7950--7964, 2021.

\bibitem[BHS23]{BrownHS23}
Gavin Brown, Samuel~B. Hopkins, and Adam~D. Smith.
\newblock Fast, sample-efficient, affine-invariant private mean and covariance estimation for subgaussian distributions.
\newblock In {\em The Thirty Sixth Annual Conference on Learning Theory, {COLT} 2023}, volume 195 of {\em Proceedings of Machine Learning Research}, pages 5578--5579. {PMLR}, 2023.

\bibitem[BKSW19]{bun2019private}
Mark Bun, Gautam Kamath, Thomas Steinke, and Steven~Z Wu.
\newblock Private hypothesis selection.
\newblock {\em Advances in Neural Information Processing Systems}, 32, 2019.

\bibitem[BS16]{BunS16}
Mark Bun and Thomas Steinke.
\newblock Concentrated differential privacy: Simplifications, extensions, and lower bounds.
\newblock In {\em Theory of Cryptography - 14th International Conference, {TCC} 2016-B, Proceedings, Part {I}}, volume 9985 of {\em Lecture Notes in Computer Science}, pages 635--658, 2016.

\bibitem[BST14]{BassilyST14}
Raef Bassily, Adam~D. Smith, and Abhradeep Thakurta.
\newblock Private empirical risk minimization: Efficient algorithms and tight error bounds.
\newblock In {\em 55th {IEEE} Annual Symposium on Foundations of Computer Science, {FOCS} 2014}, pages 464--473. {IEEE} Computer Society, 2014.

\bibitem[Bub15]{Bubeck15}
S{\'{e}}bastien Bubeck.
\newblock Convex optimization: Algorithms and complexity.
\newblock {\em Foundations and Trends in Machine Learning}, 8(3-4):231--357, 2015.

\bibitem[CCGT24]{choquette2024optimal}
Christopher~A Choquette-Choo, Arun Ganesh, and Abhradeep Thakurta.
\newblock Optimal rates for $ o (1) $-smooth dp-sco with a single epoch and large batches.
\newblock {\em arXiv preprint arXiv:2406.02716}, 2024.

\bibitem[CJJ{\etalchar{+}}23]{CarmonJJLLST23}
Yair Carmon, Arun Jambulapati, Yujia Jin, Yin~Tat Lee, Daogao Liu, Aaron Sidford, and Kevin Tian.
\newblock Resqueing parallel and private stochastic convex optimization.
\newblock In {\em 64th {IEEE} Annual Symposium on Foundations of Computer Science, {FOCS} 2023}, pages 2031--2058. {IEEE}, 2023.

\bibitem[CKM{\etalchar{+}}21]{CohenKMST21}
Edith Cohen, Haim Kaplan, Yishay Mansour, Uri Stemmer, and Eliad Tsfadia.
\newblock Differentially-private clustering of easy instances.
\newblock In {\em Proceedings of the 38th International Conference on Machine Learning, {ICML} 2021}, volume 139 of {\em Proceedings of Machine Learning Research}, pages 2049--2059. {PMLR}, 2021.

\bibitem[CLM{\etalchar{+}}16]{cohen2016geometric}
Michael~B Cohen, Yin~Tat Lee, Gary Miller, Jakub Pachocki, and Aaron Sidford.
\newblock Geometric median in nearly linear time.
\newblock In {\em Proceedings of the forty-eighth annual ACM symposium on Theory of Computing}, pages 9--21, 2016.

\bibitem[CM08]{chaudhuri2008privacy}
Kamalika Chaudhuri and Claire Monteleoni.
\newblock Privacy-preserving logistic regression.
\newblock {\em Advances in neural information processing systems}, 21, 2008.

\bibitem[CWZ21]{cai2021cost}
T~Tony Cai, Yichen Wang, and Linjun Zhang.
\newblock The cost of privacy: Optimal rates of convergence for parameter estimation with differential privacy.
\newblock {\em The Annals of Statistics}, 49(5):2825--2850, 2021.

\bibitem[DFM{\etalchar{+}}20]{du2020differentially}
Wenxin Du, Canyon Foot, Monica Moniot, Andrew Bray, and Adam Groce.
\newblock Differentially private confidence intervals.
\newblock {\em arXiv preprint arXiv:2001.02285}, 2020.

\bibitem[DR14]{dr14}
Cynthia Dwork and Aaron Roth.
\newblock The algorithmic foundations of differential privacy.
\newblock {\em Foundations and Trends{\textregistered} in Theoretical Computer Science}, 9(3--4):211--407, 2014.

\bibitem[FKT20]{FeldmanKT20}
Vitaly Feldman, Tomer Koren, and Kunal Talwar.
\newblock Private stochastic convex optimization: optimal rates in linear time.
\newblock In {\em Proceedings of the 52nd Annual {ACM} {SIGACT} Symposium on Theory of Computing, {STOC} 2020}, pages 439--449. {ACM}, 2020.

\bibitem[GLL22]{gopi2022private}
Sivakanth Gopi, Yin~Tat Lee, and Daogao Liu.
\newblock Private convex optimization via exponential mechanism.
\newblock In {\em Conference on Learning Theory}, pages 1948--1989. PMLR, 2022.

\bibitem[GLL{\etalchar{+}}23]{gopi2023private}
Sivakanth Gopi, Yin~Tat Lee, Daogao Liu, Ruoqi Shen, and Kevin Tian.
\newblock Private convex optimization in general norms.
\newblock In {\em Proceedings of the 2023 Annual ACM-SIAM Symposium on Discrete Algorithms (SODA)}, pages 5068--5089. SIAM, 2023.

\bibitem[HSU24]{haghifam2024private}
Mahdi Haghifam, Thomas Steinke, and Jonathan Ullman.
\newblock Private geometric median.
\newblock {\em Advances in Neural Information Processing Systems}, 37:46254--46293, 2024.

\bibitem[KDH23]{kuditipudi2023pretty}
Rohith Kuditipudi, John Duchi, and Saminul Haque.
\newblock A pretty fast algorithm for adaptive private mean estimation.
\newblock In {\em The Thirty Sixth Annual Conference on Learning Theory}, pages 2511--2551. PMLR, 2023.

\bibitem[KJ16]{kasiviswanathan2016efficient}
Shiva~Prasad Kasiviswanathan and Hongxia Jin.
\newblock Efficient private empirical risk minimization for high-dimensional learning.
\newblock In {\em International Conference on Machine Learning}, pages 488--497. PMLR, 2016.

\bibitem[KLL21]{kulkarni2021private}
Janardhan Kulkarni, Yin~Tat Lee, and Daogao Liu.
\newblock Private non-smooth erm and sco in subquadratic steps.
\newblock {\em Advances in Neural Information Processing Systems}, 34:4053--4064, 2021.

\bibitem[KLSU19]{KamathLSU19}
Gautam Kamath, Jerry Li, Vikrant Singhal, and Jonathan~R. Ullman.
\newblock Privately learning high-dimensional distributions.
\newblock In {\em Conference on Learning Theory, {COLT} 2019, 25-28 June 2019}, volume~99 of {\em Proceedings of Machine Learning Research}, pages 1853--1902. {PMLR}, 2019.

\bibitem[KST12]{kifer2012private}
Daniel Kifer, Adam Smith, and Abhradeep Thakurta.
\newblock Private convex empirical risk minimization and high-dimensional regression.
\newblock In {\em Conference on Learning Theory}, pages 25--1. JMLR Workshop and Conference Proceedings, 2012.

\bibitem[KV17]{karwa2017finite}
Vishesh Karwa and Salil Vadhan.
\newblock Finite sample differentially private confidence intervals.
\newblock {\em arXiv preprint arXiv:1711.03908}, 2017.

\bibitem[LKJO22]{LiuKJO22}
Xiyang Liu, Weihao Kong, Prateek Jain, and Sewoong Oh.
\newblock {DP-PCA:} statistically optimal and differentially private {PCA}.
\newblock In {\em Advances in Neural Information Processing Systems 35: Annual Conference on Neural Information Processing Systems 2022}, 2022.

\bibitem[LR91]{LopuhaaR91}
Hendrik~P.\ Lopuhaa and Peter~J.\ Rousseuw.
\newblock Breakdown points of affine equivariant estimators of multivariate location and covariance matrices.
\newblock {\em The Annals of Statistics}, 19(1):229--248, 1991.

\bibitem[Mir17]{Mironov17}
Ilya Mironov.
\newblock R{\'{e}}nyi differential privacy.
\newblock In {\em 30th {IEEE} Computer Security Foundations Symposium, {CSF} 2017}, pages 263--275. {IEEE} Computer Society, 2017.

\bibitem[NRS07]{NissimRS07}
Kobbi Nissim, Sofya Raskhodnikova, and Adam~D. Smith.
\newblock Smooth sensitivity and sampling in private data analysis.
\newblock In {\em Proceedings of the 39th Annual {ACM} Symposium on Theory of Computing, San Diego, California, USA, June 11-13, 2007}, pages 75--84. {ACM}, 2007.

\bibitem[NSV16]{NissimSV16}
Kobbi Nissim, Uri Stemmer, and Salil~P. Vadhan.
\newblock Locating a small cluster privately.
\newblock In Tova Milo and Wang{-}Chiew Tan, editors, {\em Proceedings of the 35th {ACM} {SIGMOD-SIGACT-SIGAI} Symposium on Principles of Database Systems, {PODS} 2016}, pages 413--427. {ACM}, 2016.

\bibitem[Roc76]{Rockafellar76}
R.\~Tyrrell Rockafellar.
\newblock Monotone operators and the proximal point algorithm.
\newblock {\em SIAM J.\ Control and Optimization}, 14(5):877--898, 1976.

\bibitem[TCK{\etalchar{+}}22]{tsfadia2022friendlycore}
Eliad Tsfadia, Edith Cohen, Haim Kaplan, Yishay Mansour, and Uri Stemmer.
\newblock Friendlycore: Practical differentially private aggregation.
\newblock In {\em International Conference on Machine Learning}, pages 21828--21863. PMLR, 2022.

\bibitem[Web29]{weber1929theory}
Alfred Weber.
\newblock {\em Theory of the Location of Industries}.
\newblock University of Chicago Press, 1929.

\bibitem[ZTC22]{zhang2022differentially}
Qinzi Zhang, Hoang Tran, and Ashok Cutkosky.
\newblock Differentially private online-to-batch for smooth losses.
\newblock In {\em NeurIPS}, 2022.

\end{thebibliography}
